\documentclass[12pt]{article}
\usepackage{fullpage,cmu-titlepage2}
\usepackage{times}

\ifx\pdftexversion\undefined
  \usepackage[dvips]{graphics}
\else
  \usepackage[pdftex]{graphics}
\fi

\usepackage{amsmath}
\usepackage{rotating}
\usepackage{tabularx}
\usepackage[
pageanchor=true,
plainpages=false,
pdfpagelabels, 
bookmarks,
bookmarksnumbered,
pdfborder=0 0 0,
pdfpagemode=UseOutlines]{hyperref}
\usepackage{subfigure}
\usepackage{amsmath}
\usepackage{latexsym,amsthm,amssymb,amscd,url,enumerate}
\usepackage{mathpartir,mathtools}
\usepackage[show]{ed}
\usepackage[all]{xy}
\usepackage{tikz}
\usepackage{tikz-cd}
\usetikzlibrary{arrows}

\newcommand{\id}[3][]{#2 =_{#1} #3}
\newcommand{\idtype}[3][]{\mathsf{Id}_{#1}(#2,#3)}

\newcommand{\refl}[2]{\mathsf{refl}_{#1}(#2)}
\newcommand{\reflsym}{\mathsf{refl}}
\newcommand{\defeq}{\coloneqq}
\newcommand{\sm}[1]{\Sigma_{#1}}
\newcommand{\prd}[1]{\Pi_{#1}}
\newcommand{\lam}[1]{\lambda_{#1}}
\newcommand{\U}{\mathcal{U}}

\newcommand{\iscontr}[1]{\mathsf{iscontr}(#1)}
\newcommand{\isprop}[1]{\mathsf{isprop}(#1)}
\newcommand{\hfiber}[2]{\mathsf{hfiber}_{#1}(#2)}
\newcommand{\iseq}[1]{\mathsf{iseq}(#1)}
\newcommand{\nat}{\mathbb{N}}
\newcommand{\one}{\mathbf{1}}
\newcommand{\two}{\mathbf{2}}
\newcommand{\ap}[2]{\mathsf{ap}_{#1}(#2)}
\newcommand{\J}[2]{\mathsf{J}_{#1,#2}}
\newcommand{\opp}[1]{\mathord{{#1}^{-1}}}
\newcommand{\trans}[2]{{#2}_*^{#1}}
\newcommand{\transc}[2]{{#2}^*_{#1}}
\newcommand{\idfun}[1]{\mathsf{id}_{#1}}
\newcommand{\comp}{\circ}
\newcommand{\dap}[2]{\mathsf{dap}_{#1}(#2)}
\newcommand{\dapc}[2]{\mathsf{dap}^{#1}(#2)}
\newcommand{\pair}[2]{^{=}\mathsf{E}^{\times}_{#1,#2}}
\newcommand{\dpair}[2]{^{=}\mathsf{E}^{\Sigma}_{#1,#2}}
\newcommand{\dpath}[2]{^{\Sigma}\mathsf{E}^{=}_{#1,#2}}
\newcommand{\dpairsym}{^{=}\mathsf{E}^{\Sigma}}
\newcommand{\dpathsym}{^{\Sigma}\mathsf{E}^{=}}
\newcommand{\happly}[2]{^{=}\mathsf{E}^{\Pi}_{#1,#2}}
\newcommand{\ideq}[2]{^{=}\mathsf{E}^{\simeq}_{#1,#2}}
\newcommand{\fst}[1]{\mathsf{fst}(#1)}
\newcommand{\snd}[1]{\mathsf{snd}(#1)}
\newcommand{\fstsym}{\mathsf{fst}}

\newcommand{\point}{\mathsf{pt}}
\newcommand{\cell}{\mathsf{cl}}
\newcommand{\Sn}[1]{\mathbb{S}}
\newcommand{\Sna}[1]{\mathbb{S}_a}
\newcommand{\base}{\mathsf{base}}
\newcommand{\lp}{\mathsf{loop}}
\newcommand{\north}{\mathsf{north}}
\newcommand{\south}{\mathsf{south}}
\newcommand{\w}{\mathsf{west}}
\newcommand{\e}{\mathsf{east}}
\newcommand{\trunc}[1]{||#1||}
\newcommand{\inj}[1]{|#1|}
\newcommand{\squash}{\mathsf{sq}}
\newcommand{\circrec}[3]{\mathsf{rec}^{\Sn{1}}_{#1,#2,#3}}

\newcommand{\truncrec}[3]{\mathsf{rec}^{\trunc{A}}_{#1,#2,#3}}
\newcommand{\truncind}[3]{\mathsf{ind}^{\trunc{A}}_{#1,#2,#3}}
\newcommand{\wsusprec}[3]{\mathsf{rec}^{\wsusp{A}{f}{g}}_{#1,#2,#3}}
\newcommand{\wsuspind}[3]{\mathsf{ind}^{\wsusp{A}{f}{g}}_{#1,#2,#3}}
\newcommand{\circind}[3]{\mathsf{ind}^{\Sn{1}}_{#1,#2,#3}}

\newcommand{\circaind}[5]{\mathsf{ind}^{\Sna{1}}_{#1,#2,#3,#4,#5}}
\newcommand{\hascircind}[1]{\mathsf{has}\mbox{-}{\Sn{1}}\mbox{-}\mathsf{ind}_{#1}}
\newcommand{\hascircaind}[1]{\mathsf{has}\mbox{-}{\Sna{1}}\mbox{-}\mathsf{ind}_{#1}}
\newcommand{\hastruncrec}[1]{\mathsf{has}\mbox{-}{\trunc{A}}\mbox{-}\mathsf{rec}_{#1}}
\newcommand{\hastruncind}[1]{\mathsf{has}\mbox{-}{\trunc{A}}\mbox{-}\mathsf{ind}_{#1}}
\newcommand{\haswsusprec}[1]{\mathsf{has}\mbox{-}{\wsusp{A}{f}{g}}\mbox{-}\mathsf{rec}_{#1}}
\newcommand{\haswsuspind}[1]{\mathsf{has}\mbox{-}{\wsusp{A}{f}{g}}\mbox{-}\mathsf{ind}_{#1}}
\newcommand{\circalgsym}{\Sn{1}\mbox{-}\mathsf{alg}}
\newcommand{\truncalgsym}{\trunc{A}\mbox{-}\mathsf{alg}}
\newcommand{\truncfibalgsym}{\trunc{A}\mbox{-}\mathsf{fib}\mbox{-}\mathsf{alg}}
\newcommand{\truncalg}[1]{\truncalgsym_{#1}}
\newcommand{\truncfibalg}[1]{\truncfibalgsym_{#1}}
\newcommand{\trunchom}{\trunc{A}\mbox{-}\mathsf{hom}}
\newcommand{\wsusphom}{\wsusp{A}{f}{g}\mbox{-}\mathsf{hom}}
\newcommand{\wsuspfibhom}{\wsusp{A}{f}{g}\mbox{-}\mathsf{fib}\mbox{-}\mathsf{hom}}
\newcommand{\truncfibhom}{\trunc{A}\mbox{-}\mathsf{fib}\mbox{-}\mathsf{hom}}
\newcommand{\circaalgsym}{\Sna{1}\mbox{-}\mathsf{alg}}
\newcommand{\circaalg}[1]{\circaalgsym_{#1}}
\newcommand{\wsusp}[3]{\W_{#1}^{#2,#3}}
\newcommand{\recsym}{\mathsf{rec}}

\newcommand{\hinitial}[1]{\mathsf{is}\mbox{-}\wsusp{A}{f}{g}\mbox{-}\mathsf{hinit}_{#1}}
\newcommand{\hasuniq}[1]{\mathsf{has}\mbox{-}{\wsusp{A}{f}{g}}\mbox{-}\mathsf{rec}\mbox{-}\mathsf{uniq}_{#1}}
\newcommand{\hasinduniq}[1]{\mathsf{has}\mbox{-}{\wsusp{A}{f}{g}}\mbox{-}\mathsf{ind}\mbox{-}\mathsf{uniq}_{#1}}
\newcommand{\circalg}[1]{\circalgsym_{#1}}
\newcommand{\wsuspalg}[1]{\wsusp{A}{f}{g}\mbox{-}\mathsf{alg}_{#1}}
\newcommand{\wsuspcoh}{\wsusp{A}{f}{g}\mbox{-}\mathsf{coh}}
\newcommand{\wsuspfibcoh}{\wsusp{A}{f}{g}\mbox{-}\mathsf{fib}\mbox{-}\mathsf{coh}}
\newcommand{\wsuspfibalg}[1]{\wsusp{A}{f}{g}\mbox{-}\mathsf{fib}\mbox{-}\mathsf{alg}_{#1}}
\newcommand{\circalgtocircaalg}{\Sn{1}\mbox{-}\mathsf{to}\mbox{-}\circaalgsym_{\U_i}}
\newcommand{\circaalgtocircalg}{\Sna{1} \mbox{-}\mathsf{to}\mbox{-}\circalgsym_{\U_i}}
\newcommand{\circalgtowsuspalg}{\Sn{1}\mbox{-}\mathsf{to}\mbox{-}\wsusp{A}{f}{g}\mbox{-}\mathsf{alg}_{\U_i}}
\newcommand{\wsuspalgtocircalg}{\wsusp{A}{f}{g}\mbox{-}\mathsf{to}\mbox{-}\Sn{1}\mbox{-}\mathsf{alg}_{\U_i}}
\newcommand{\circaalgtowsuspalg}{\Sna{1}\mbox{-}\mathsf{to}\mbox{-}\wsusp{A}{f}{g}\mbox{-}\mathsf{alg}_{\U_i}}
\newcommand{\wsuspalgtocircaalg}{\wsusp{A}{f}{g}\mbox{-}\mathsf{to}\mbox{-}\Sna{1}\mbox{-}\mathsf{alg}_{\U_i}}
\newcommand{\poinfun}[4]{\mathfrak{Pf}_{#1,#2,#3,#4}}
\newcommand{\poinfuncoh}[2]{\mathfrak{Pf}\mbox{-}\mathsf{coh}_{#1,#2}}
\newcommand{\W}{\mathsf{W}}
\newcommand{\invtri}{\mathbf{I}_\triangle}
\newcommand{\invsq}{\mathbf{I}_\square}

\newcommand{\twocell}{\wsusp{A}{f}{g}\mbox{-}\mathit{2}\mbox{-}\mathsf{cell}}
\newcommand{\twofibcell}{\wsusp{A}{f}{g}\mbox{-}\mathit{2}\mbox{-}\mathsf{fib}\mbox{-}\mathsf{cell}}
\newcommand{\ct}{%
  \mathchoice{\mathbin{\raisebox{0.5ex}{$\displaystyle\centerdot$}}}%
             {\mathbin{\raisebox{0.5ex}{$\centerdot$}}}%
             {\mathbin{\raisebox{0.25ex}{$\scriptstyle\,\centerdot\,$}}}%
             {\mathbin{\raisebox{0.1ex}{$\scriptscriptstyle\,\centerdot\,$}}}}

\newtheorem{theorem}{Theorem}						
\newtheorem{definition}[theorem]{Definition}
\newtheorem{proposition}[theorem]{Proposition}
\newtheorem{axiom}{Axiom}
\newtheorem{notation}[theorem]{Notation}
\newtheorem{corollary}[theorem]{Corollary}

\title{Higher Inductive Types as Homotopy-Initial Algebras}

\author{Kristina Sojakova}

\date{January 2014}

\abstract{Homotopy Type Theory is a new field of mathematics based on the surprising and elegant correspondence between Martin-L¨of’s constructive type theory and abstract homotopy theory. We have a powerful interplay between these disciplines - we can use geometric intuition to formulate new concepts in type theory and, conversely, use type-theoretic machinery to verify and often simplify existing mathematical proofs. A crucial ingredient in this new system are higher inductive types, which allow us to represent objects such as spheres, tori, pushouts, and quotients. We investigate a variant of higher inductive types whose computational behavior is determined up to a higher path. We show that in this setting, higher inductive types are characterized by the universal property of being a homotopy-initial algebra.}


\trnumber{CMU-CS-14-101}

\support{Support for this research was provided by the Funda\c{c}\~{a}o para a
Ci\^{e}ncia e a Tecnologia (Portuguese Foundation for Science and Technology) through the Carnegie Mellon Portugal Program under grant NGN-44 and by the National Science Foundation Grant DMS-1001191.}

\begin{document}
\renewcommand*{\thepage}{title-\arabic{page}} 
\maketitle
\renewcommand*{\thepage}{\arabic{page}} 

\section{Introduction}\label{intro}
Homotopy Type Theory (HoTT) has recently generated significant interest among type theorists and mathematicians alike. It uncovers deep connections between Martin-L{\"o}f's dependent type theory (\cite{martin_loef,martin_loef_b}) and the fields of abstract homotopy theory, higher categories, and algebraic topology (\cite{awodey_warren,gambino_garner,garner,hofmann_streicher,ssets,licata_harper,lumsdaine,berg_garner_b,berg_garner,uf_talk,warren_phd}). Insights from homotopy theory are used to add new concepts to the type theory, such as the representation of various geometric objects as higher inductive types. Conversely, type theory is used to formalize and verify existing mathematical proofs using proof assistants such as Coq \cite{coq} and Agda \cite{agda}. Moreover, type-theoretic insights often help us discover novel proofs of known results which are simpler than their homotopy-theoretic versions: the calculation of $\pi_n(\mathbf{S}^n)$ (\cite{licata_shulman,licata_brunerie}); the Freudenthal Suspension Theorem \cite{hott}; the Blakers-Massey Theorem \cite{hott}, etc.

As a formal system, HoTT \cite{hott} is a generalization of intensional Martin-L{\"o}f Type Theory with two features motivated by abstract homotopy theory: Voevodsky's \emph{univalence axiom} (\cite{ssets,uf_talk}) and \emph{higher-inductive types} (\cite{lumsdaine_blog,shulman_blog}). The slogan in HoTT is that \emph{types are topological spaces, terms are points, and proofs of identity are paths between points}. The structure of an identity type in HoTT is thus far more complex than just being the ``least reflexive relation", despite the definition of $\idtype[A]{M}{N}$ as an inductive type with a single constructor $\refl{A}{M} : \idtype[A]{M}{M}$. It is a beautiful, and perhaps surprising, fact that not only does this richer theory admit an interpretation into homotopy theory (\cite{awodey_warren}, \cite{ssets}) but that many fundamental concepts and results from mathematics arise naturally as constructions and theorems of HoTT.

For example, the unit circle $\mathbf{S}^1$ is defined as a higher inductive type with a fixed point $\base$ and a loop $\lp$ based at $\base$. It comes with a recursion principle which says that to construct a function $f : \mathbf{S}^1 \to X$, it suffices to supply a point $x:X$ and a loop based at $x$. The value $f(\base)$ then computes to $x$. Such definitional computation rules are convenient to work with but also pose a number of problems. For instance, an alternative encoding of the circle as a higher inductive type $\mathbf{S}^1_a$ specifies two fixed points $\south$, $\north$ and two paths from $\north$ to $\south$, called $\e$ and $\w$. The recursion principle then says that in order to construct a function $f : \mathbf{S}^1_a \to X$, it suffices to supply two points $x,y:X$ and two paths between them. The values $f(\north)$ and $f(\south)$ the compute to $x$ and $y$ respectively.

We have a natural way of relating these two representations: in one direction, map $\base$ to $\north$ and $\lp$ to $\e$; in the other direction, map both $\north$ and $\south$ to $\base$ and map $\e$ to $\lp$ and $\w$ to the identity path at $\base$. Unfortunately, the circles $\mathbf{S}^1$, $\mathbf{S}_a^1$ related this way do \emph{not} satisfy the same definitional equalities, which poses a compatibility issue. Yet more severe problem arises with the propositional truncation $\trunc{A}$ (also known as a bracket or squash type, see \cite{awodey_bauer}). The job of the truncated type $\trunc{A}$ is to provide evidence that $A$ is inhabited without giving up the particular witness term $a:A$. However, \cite{trunc_inverse} shows that under some conditions, the inhabitant $a$ can be recovered. A definitional computation rule is again the culprit.

In this paper we thus consider higher inductive types endowed with \emph{propositional} computation rules: in the case of $\mathbf{S}^1$ such a rule would state that there is a \emph{path} between $f(\base)$ and $x:X$. Types in this setting tend to keep many of the desirable properties; for instance, it can be shown that the main result of \cite{licata_shulman}, that the fundamental group of the circle is the group of integers, carries over to the case when \emph{both} the circle and the integer types have propositional computational behavior. In addition, we can now show that higher inductive types are characterized by the universal property of being a \emph{homotopy-initial algebra}. This notion was first introduced in \cite{wtypes}, where an analogous result was established for the ``ordinary" inductive type of well-founded trees (Martin-L{\"o}f's W-types). In the higher-dimensional setting, an \emph{algebra} is a type $X$ together with a number of finitary operations $f,g,h \ldots$, which are allowed to act not only on $X$ but also on any higher identity type over $X$. An \emph{algebra homomorphism} has to preserve all operations up to a higher homotopy. Finally, an algebra $\mathcal{X}$ is \emph{homotopy-initial} if the type of homomorphisms from $\mathcal{X}$ to any other algebra $\mathcal{Y}$ is contractible.

Our main theorem is stated for a class of higher inductive types which we call $\W$-suspensions; they generalize the types $\mathbf{S}^1$, $\mathbf{S}_a^1$, and others by allowing any number of fixed points and any number of paths between any two of these points. We show that the induction principle for $\W$-suspensions is equivalent to the recursion principle plus a certain uniqueness condition, which in turn is shown to be equivalent to homotopy-initiality. This extends the main result of \cite{wtypes} for ``ordinary'' inductive types to the important, and much more difficult, higher-dimensional case. 

\section{Basic Homotopy Type Theory}\label{prelim}

The core of HoTT is a dependent type theory with 
\begin{itemize}

\item dependent pair types $\sm{x:A}{B(x)}$ and dependent function types $\prd{x:A}{B(x)}$ (with the non-dependent versions $A \times B$ and $A \to B$). To stay consistent with the presentation in \cite{hott}, we assume definitional $\eta$-conversion for functions but do not assume it for pairs.

\item a cumulative hierarchy of universes $\U _{0} : \U_1 : \U_2 : \ldots$ in the style of Russell.

\item intensional identity types $\idtype[A]{M}{N}$, also denoted by $\id[A]{M}{N}$. We have the usual formation and introduction rules; the elimination and computation rules are recalled below:
\begin{mathpar}
\inferrule{E : \prd{x,y : A} \idtype[A]{x}{y} \to \U_i \\ d : \prd{x : A} E(x,x,\refl{A}{x})}
          {\J{E}{d} : \prd{x,y : A}\prd{p : \idtype[A]{M}{N}} E(x,y,p)}

\inferrule{E : \prd{x,y : A} \idtype[A]{x}{y} \to \U_i \\ d : \prd{x : A} E(x,x,\refl{A}{x}) \\ M : A}
          {\J{E}{d}(M,M,\refl{A}{M}) \equiv d(M) : E(M,M,\refl{A}{M})}
\end{mathpar}
These rules are, of course, applicable in any context $\Gamma$; we follow the standard convention of omitting it. If the type $\idtype[A]{M}{N}$ is inhabited, we call $M$ and $N$ \emph{propositionally equal}. If we do not care about the specific equality witness, we often simply say that $\id[A]{M}{N}$ or if the type $A$ is clear, $\id{M}{N}$. A term $p : \id[A]{M}{N}$ will be often called a \emph{path} and the process of applying the identity elimination rule will be referred to as \emph{path induction}. Definitional equality between $M,N:A$ will be denoted as $M \equiv N : A$.
\end{itemize}
We emphasize that apart from the aforementioned identity rules, univalence, and higher inductive types there are no other rules governing the behavior of identity types - in particular, we assert neither any form of Streicher's K-rule \cite{streichersk} nor the identity reflection rule.

The rest of this section describes the univalence axiom and some key properties of identity types; higher inductive types are discussed in Section.~\ref{hits}. For a thorough exposition of homotopy type theory we refer the reader to \cite{hott}.

\subsection{Groupoid laws}
Proofs of identity behave much like paths in topological spaces: they can be reversed, concatenated, mapped along functions, etc. Below we summarize a few of these properties:
\begin{itemize}
\item For any path $p : \id[A]{x}{y}$ there is a path $\opp{p} : \id[A]{y}{x}$, and we have $\opp{\refl{A}{x}} \equiv \refl{A}{x}$.
\item For any paths $p : \id[A]{x}{y}$ and $q : \id[A]{y}{z}$ there is a path $p \ct q : \id[A]{x}{z}$, and we have
$\refl{A}{x} \ct \refl{A}{x} \equiv \refl{A}{x}$.
\item Associativity of composition: for any paths $p : \id[A]{x}{y}$, $q : \id[A]{y}{z}$, and $r : \id[A]{y}{u}$ we have
$\id{(p \ct q) \ct r}{p \ct (q \ct r)}$.
\item We have $\refl{A}{x} \ct p = p$ and $p \ct \refl{A}{y} = p$ for any $p : \id[A]{x}{y}$.
\item We have $p \ct \opp{p} = \refl{A}{x}$, $\opp{p} \ct p = \refl{A}{y}$, and $\opp{(\opp{p})} = p$, $\opp{(p \ct q)} = \opp{q} \ct \opp{p}$ for any $p : \id[A]{x}{y}$, $q : \id[A]{y}{z}$.
\item For any type family $P : A \to \U_i$ and path $p : \id[A]{x}{y}$ there are functions $\trans{P}{p} : P(x) \to P(y)$ and $\transc{P}{p} : P(y) \to P(x)$, called the \emph{covariant transport} and \emph{contravariant transport}, respectively. We furthermore have $\trans{P}{\refl{A}{x}} \equiv \transc{P}{\refl{A}{x}} \equiv \idfun{P(x)}$.
\item We have $\trans{P}{(\opp{p})} = \transc{P}{p}$, $\transc{P}{(\opp{p})} = \trans{P}{p}$ and $\trans{P}{(p \ct q)} = \trans{P}{q} \comp \trans{P}{p}$, $\transc{P}{(p \ct q)} = \transc{P}{p} \comp \transc{P}{q}$ for any family $P : A \to \U_i$ and paths $p : \id[A]{x}{y}$, $q : \id[A]{y}{z}$.
\item For any function $f : A \to B$ and path $p : \id[A]{x}{y}$, there is a path $\ap{f}{p} : \id[B]{f(x)}{f(y)}$ and we have
$\ap{f}{\refl{A}{x}} \equiv \refl{B}{f(x)}$. 
\item We have $\ap{f}{\opp{p}} = \opp{\ap{f}{p}}$ and $\ap{f}{p \ct q} = \ap{f}{p} \ct \ap{f}{q}$ for any $f : A \to B$ and $p : \id[A]{x}{y}$, $q : \id[A]{y}{z}$.
\item We have $\ap{g\circ f}{p} = \ap{g}{\ap{f}{p}}$ for any $f : A \to B$, $g : B \to C$ and $p : \id[A]{x}{y}$.
\item For a dependent function $f : \prd{x:A} B(x)$ and path $p : \id[A]{x}{y}$, there are paths $\dap{f}{p} : \id[B(y)]{\trans{B}{p}(f(x))}{f(y)}$ and $\dapc{f}{p} : \id[B(x)]{\transc{B}{p}(f(y))}{f(x)}$. We also have $\dap{f}{\refl{A}{x}} \equiv \dapc{f}{\refl{A}{x}} \equiv \refl{B(x)}{f(x)}$.
\item All constructs respect propositional equality.
\end{itemize}

\subsection{Homotopies between functions}
A homotopy between two functions $f,g$ is a ``natural transformation" between $f$ and $g$:
\begin{definition}
For $f,g : \prd{x:A} B(x)$, we define the type \[f \sim g \defeq \prd{a:A} \id[B(a)]{(f(a)}{g(a))}\] and call it the \emph{type of homotopies} between $f$ and $g$.
\end{definition}
\begin{proposition}
For any $f,g : A \to B$, $\alpha : f \sim g$, and path $p : \id[A]{x}{y}$, the diagram on the left commutes:
\begin{center}
\begin{tikzpicture}
\node (N0) at (1.5,0.75) {=};
\node (N1) at (0,1.5) {$f(x)$};
\node (N2) at (3,1.5) {$f(y)$};
\node (N3) at (0,0) {$g(x)$};
\node (N4) at (3,0) {$g(y)$};
\draw[-] (N1) -- node[above]{\footnotesize $\ap{f}{p}$} (N2);
\draw[-] (N1) -- node[left]{\footnotesize $\alpha(x)$} (N3);
\draw[-] (N2) -- node[right]{\footnotesize $\alpha(y)$} (N4);
\draw[-] (N3) -- node[below]{\footnotesize $\ap{g}{p}$} (N4);

\node (N10) at (8.5,0.75) {=};
\node (N11) at (7,1.5) {$\trans{B}{p}(f(x))$};
\node (N12) at (10,1.5) {$f(y)$};
\node (N13) at (7,0) {$\trans{B}{p}(g(x))$};
\node (N14) at (10,0) {$g(y)$};
\draw[-] (N11) -- node[above]{\footnotesize $\dap{f}{p}$} (N12);
\draw[-] (N11) -- node[left]{\footnotesize $\ap{\trans{B}{p}}{\alpha(x)}$} (N13);
\draw[-] (N12) -- node[right]{\footnotesize $\alpha(y)$} (N14);
\draw[-] (N13) -- node[below]{\footnotesize $\dap{g}{p}$} (N14);
\end{tikzpicture}
\end{center}
This property will be referred to as the \emph{naturality} of $\alpha$. We likewise have a dependent version of naturality when $f,g:\prd{x:A} B(x)$, which is shown on the right.
\end{proposition}
Of course, there is also a contravariant version of dependent naturality, which we will not need.

\subsection{Truncation levels}
In general, the structure of paths on a type $A$ can be highly nontrivial - we can have many distinct \emph{0-cells} $x,y,\ldots : A$; there can be many distinct \emph{1-cells} $p,q,\ldots : \id[A]{x}{y}$; there can be many distinct \emph{2-cells} $\gamma,\delta,\ldots : \id[{\id[A]{x}{y}}]{p}{q}$; ad infinitum. The hierarchy of truncation levels describes those types which are, informally speaking, trivial beyond a certain dimension: a type $A$ of truncation level $n$ can be characterized by the property that all $m$-cells for $m > n$ with the same source and target are equal. From this intuitive description we can see that the hierarchy is cumulative.
 
It is customary to also speak of truncation levels $-2$ and $-1$, called \emph{contractible types} and \emph{mere propositions} respectively:
\begin{definition}
A type $A:\U_i$ is called \emph{contractible} if there exists a point $a : A$ such that any other point $x : A$ is equal to $a$: \[\iscontr{A} \defeq \sm{a:A}\prd{x:A} (\id[A]{a}{x})\]
A type $A:\U_i$ is called a \emph{mere proposition} if all its inhabitants are equal: \[\isprop{A} \defeq \prd{x,y:A} (\id[A]{x}{y})\]
\end{definition}
Thus, a contractible type can be seen as having exactly one inhabitant, up to equality; a mere proposition can be seen as having at most one inhabitant, up to equality. Clearly:
\begin{proposition}
If $A$ is contractible then $A$ is a mere proposition.
\end{proposition}
The existence of a path between any two points implies more than just path-connectedness:
\begin{proposition}\label{thm_contr_path}
If $A$ is a mere proposition, then $\id[A]{x}{y}$ is contractible for any $x,y:A$.
\end{proposition}
Thus, contractible types are in a sense the ``nicest" possible: any two points are equal up to a \emph{1-cell}, which itself is unique up to a \emph{2-cell}, which itself is unique up to a \emph{3-cell}, and so on. Mere propositions are the ``nicest" ones after contractible spaces. We can now easily show:
\begin{corollary}\label{thm_contr_prop}
For any $A$, $\iscontr{A}$ and $\isprop{A}$ are mere propositions.
\end{corollary}

\subsection{Equivalences}
A crucial concept in HoTT is that of an equivalence between types. Intuitively, we want to think of two types $A,B$ as equivalent if there exists a bijection between them, i.e., a function $f : A \to B$ such that the preimage of any single point $b:B$ under $f$ is again a single point. Phrasing this in the language of HoTT:
\begin{definition}
We define the \emph{homotopy fiber} of a function $f : A\to B$ at $b:B$ by
\[ \hfiber{f}{b} \defeq \sm{a:A} (\id[B]{f(a)}{b}) \]
\end{definition}
\begin{definition}
A function $f : A \to B$ is called an \emph{equivalence} if all its homotopy fibers are contractible: \[\iseq{f} \defeq \prd{b:B} \iscontr{\hfiber{f}{b}}\]
We define \[(A \simeq B) \defeq \sm{f:A\to B}\iseq{f}\] and call $A$ and $B$ \emph{equivalent} if the above type is inhabited.
\end{definition}
Unsurprisingly, we can prove that $A$ and $B$ are equivalent by constructing functions going back and forth, which compose to identity on both sides\footnote{Although the type of such functions itself is not equivalent to $A \simeq B$, see Chpt.~4 of \cite{hott}.}; this is also a necessary condition.
\begin{proposition}\label{thm_quasieq}
Two types $A$ and $B$ are equivalent if and only if there exist functions $f:A \to B$ and $g : B \to A$ such that $g \comp f \sim \idfun{A}$ and $f \comp g \sim \idfun{B}$.
\end{proposition}
We will refer to such functions $f$ and $g$ as forming a \emph{quasi-equivalence}.
From this we can easily show:
\begin{proposition}
Equivalence of types is an equivalence relation.
\end{proposition}
We call $A$ and $B$ \emph{logically equivalent} if there are functions $f:A\to B$, $g:B \to A$. Clearly, if both types are mere propositions then logical equivalence implies $A \simeq B$. For example:
\begin{corollary}\label{thm_contr_char}
For any $A$, $\iscontr{A} \simeq (A \times \isprop{A})$.
\end{corollary}

\subsection{Structure of path types}
Let us first consider the product type $A \times B$. We would like for two pairs $c,d:A \times B$ to be (propositionally) equal precisely when their first and second projections are equal. By path induction we can easily construct a function \[\pair{c}{d} : (\id{c}{d}) \to(\id{\fst{c}}{\fst{d}}) \times (\id{\snd{c}}{\snd{d}})\] We can show:
\begin{proposition}
The map $\pair{c}{d}$ is an equivalence for any $c,d:A\times B$.
\end{proposition}

We have a similar correspondence for dependent pairs; however, the second projections of $c,d:\sm{x:A}B(x)$ now lie in different fibers of $B$ and we employ (covariant) transport. By path induction we can define \[\dpair{c}{d} : (\id{c}{d}) \to \sm{(p:\id{\fst{c}}{\fst{d}})} (\id{\trans{B}{p}(\snd{c})}{\snd{d}}) \]
\begin{proposition}\label{thm_pair_equiv}
The map $\dpair{c}{d}$ is an equivalence for any $c,d:\sm{x:A}B(x)$.
\end{proposition}
We also have an analogous correspondence using a contravariant transport.

We would like for two types $A,B:\U_i$ to be equal precisely when they are equivalent. As before, we can easily obtain a function \[\ideq{A}{B} : (\id{A}{B}) \to (A \simeq B)\] The univalence axiom now states that this map is an equivalence:
\begin{axiom}[Univalence]
The map $\ideq{A}{B}$ is an equivalence for any $A,B:\U_i$.
\end{axiom}
It follows from univalence that \emph{equivalent types are equal} and hence they satisfy the same properties:
\begin{proposition}\label{thm_equiv_same}
For any type family $P : \U_i \to \U_j$, and types $A,B:\U_i$ with $A \simeq B$, we have that $P(A) \simeq P(B)$.
Thus in particular, $P(A)$ is inhabited precisely when $P(B)$ is.
\end{proposition}

Finally, two functions $f,g : \prd{x:A} B(x)$ should be equal precisely when there exists a homotopy between them. Constructing a map \[\happly{f}{g} : (\id{f}{g}) \to (f \sim g)\] is easy. Showing that this map is an equivalence (or even constructing a map in the opposite direction) is much harder, and is in fact among the chief consequences of univalence:
\begin{proposition}
The map $\happly{f}{g}$ is an equivalence for any $f,g:\prd{x:A}B(x)$.
\end{proposition}
\begin{proof}
See Chpt.~4.9 of \cite{hott}.
\end{proof}

\section{Higher Inductive Types}\label{hits}
An inductive type $X$ can be understood as being \emph{freely generated} by a collection of constructors: in the familiar case of natural numbers, we have the two constructors for zero and successor. The property of being freely generated can be stated as an induction principle: in order to show that a property $P : \nat \to \U_i$ holds for all $n:\nat$, it suffices to show that it holds for zero and is preserved by the successor operation. As a special case, we get the recursion principle: in order to define a map $f:\nat \to C$, is suffices to determine its value at zero and it’s behavior with respect to successor. 

Higher inductive types generalize ordinary inductive types by allowing constructors involving \emph{path spaces of $X$} rather than just $X$ itself, as the next example shows.

\subsection{The circle} The unit circle $\bf{S}^{1}$, denoted by $\Sn{1}:\U_0$, can be represented as an inductive type with two constructors \cite{licata_shulman}: 
\begin{align*}
  & \base : \Sn{1} \\
	& \lp : \id[\Sn{1}]{\base}{\base}
\end{align*}
pictured as 
\begin{center}	
\begin{tikzpicture}[scale=1.3,cap=round,>=latex]
        \draw[thick] (0cm,0cm) circle(0.7cm);
				\filldraw[black] (270:0.7cm) circle(1pt);           
				\draw (270:1cm) node[fill=white] {$\base$};				
			  \draw (0:1.2cm) node[fill=white] {$\lp$};
\end{tikzpicture}			
\end{center}
This in particular means that we have further paths, such as $\opp{\lp} \ct {\lp} \ct \lp \ct \refl{\Sn{1}}{\base}$
(which is equal to $\lp$).

We can reason about the circle using the principle of \emph{circle recursion}, also called \emph{simple elimination} for $\Sn{1}$, which tells us that in order to construct a function out of $\Sn{1}$ into a type $C$, it suffices to supply a point $c : C$ and a loop $s : \id[C]{c}{c}$.
\begin{mathpar}
\inferrule{C : \U_i \\ c : C \\ s : \id[C]{c}{c}}{\circrec{C}{c}{s} : \Sn{1} \to C}
\end{mathpar}
Furthermore, the recursor has the expected behavior on the 0-cell constructor $\base$ (we omit the premises):
\[ \circrec{C}{c}{s}(\base) \equiv c : C \]
We also have a computation rule for the 1-cell constructor $\lp$:
\[ \id[{\idtype[C]{c}{c}}]{\ap{\circrec{C}{c}{s}}{\lp}}{s} \]
This rule type-checks by virtue of the previous one. We note that in order to record the effect of the recursor on the path $\lp$, we use the ``action-on-paths" construct $\mathsf{ap}$. Since this is a derived notion rather than a primitive one, we state the rule as a propositional rather than definitional equality. 

We also have the more general principle of \emph{circle induction}, also called \emph{dependent elimination} for $\Sn{1}$, which subsumes recursion. Instead of a type $C : \U_i$ we now have a type family $E : \Sn{1} \to \U_i$. Where previously we required a $c : C$, we now need a point $e : E(\base)$. Finally, an obvious generalization of needing a loop $s : \id[C]{c}{c}$ would be to ask for a loop $d : \id[E(\base)]{e}{e}$. However, this would be incorrect: once we have our desired inductor of type $\prd{x:\Sn{1}} E(x)$, its effect on $\lp$ is not a loop at $e$ in the fiber $E(\base)$ but a path from $\trans{E}{\lp}(e)$ to $e$ in $E(\base)$ (or its contravariant version). The induction principle thus takes the following form:
\begin{mathpar}
\inferrule{E : \Sn{1} \to \U_i \\ e : E(\base) \\ d : \id[E(\base)]{\trans{E}{\lp}(e)}{e}}{\circind{E}{e}{d} : \prd{x:\Sn{1}} E(x)}
\end{mathpar}
We have the associated computation rules:
\begin{align*} \circind{E}{e}{d}(\base) & \equiv e : E(\base) \\
\id[{\idtype[E(\base)]{\trans{E}{\lp}(e)}{e}}]{\dap{\circind{E}{e}{d}}{\lp}&}{d}
\end{align*}

\subsection{The circle, round two}
We could have alternatively represented the circle as an inductive type $\Sna{1}:\U_0$ with four constructors:
\begin{align*}
  \north & : \Sna{1} \\
	\south & : \Sna{1} \\
	\e & : \id[\Sna{1}]{\north}{\south} \\
	\w & : \id[\Sna{1}]{\north}{\south}
\end{align*}
pictured as
\begin{center}	
\begin{tikzpicture}[scale=1.3,cap=round,>=latex]
        \draw[thick] (0cm,0cm) circle(0.7cm);
				\filldraw[black] (90:0.7cm) circle(1pt);
				\filldraw[black] (270:0.7cm) circle(1pt);
				\draw (90:1.1cm) node[fill=white] {$\north$};
				\draw (270:1cm) node[fill=white] {$\south$};				
			  \draw (180:1.2cm) node[fill=white] {$\e$};
				\draw (0:1.2cm) node[fill=white] {$\w$};
\end{tikzpicture}			
\end{center}
The corresponding induction principle is
\begin{mathpar}
\inferrule{E : \Sna{1} \to \U_i \\ u : E(\north) \\ v : E(\south) \\ \mu : \id[E(\south)]{\trans{E}{\e}(u)}{v} \\ \nu : \id[E(\south)]{\trans{E}{\w}(u)}{v}\;\;\;\;}{\circaind{E}{u}{v}{\mu}{\nu} : \prd{x:\Sna{1}} E(x)}
\end{mathpar}
with the associated computation rules
\begin{align*}
\circaind{E}{u}{v}{\mu}{\nu}(\north) & \equiv u : E(\north)\\
\circaind{E}{u}{v}{\mu}{\nu}(\south) & \equiv v : E(\south)
\end{align*}
and
\begin{align*}
\id{\dap{\circaind{E}{u}{v}{\mu}{\nu}}{\e}}{\mu} \\
\id{\dap{\circaind{E}{u}{v}{\mu}{\nu}}{\w}}{\nu}
\end{align*}

As expected, the two circle types are equivalent:
\begin{proposition}\label{thm_circeq}
We have $\Sn{1} \simeq \Sna{1}$.
\end{proposition}
\begin{proof}[Proof sketch]
From left to right, map $\base$ to $\north$ and $\lp$ to $\e \ct \opp{\w}$. From right to left, map both $\north$ and $\south$ to $\base$, $\e$ to $\lp$, and $\w$ to $\refl{\Sn{1}}{\base}$. Using the respective induction principles, show that these two mappings compose to identity on both sides and apply Prop.~\ref{thm_quasieq}. 
\end{proof}

\subsection{Computation laws, revisited}
Prop.~\ref{thm_circeq} together with univalence imply that the types $\Sn{1}$ and $\Sna{1}$ are equal and hence satisfy the same properties (see Prop.~\ref{thm_equiv_same}). We would thus expect the induction principle for $\Sn{1}$ to carry over to $\Sna{1}$, and vice versa. Indeed, with a little effort we can show the former:
\begin{proposition}\label{thm_snaind}
The type $\Sna{1}$ satisfies the induction and computation laws for $\Sn{1}$, with $\north$ acting as the constructor $\base$ and
$\e \ct \opp{\w}$ acting as the constructor $\lp$.
\end{proposition}
In the other direction, though, we hit a snag - the only obvious choice we have is to define both points $\north$ and $\south$ to be $\base$, one of the paths $\w$ and $\e$ to be $\lp$, and the other one the identity path at $\base$. This, however, does not give us the desired induction principle: unless the two given points $u : E(\base)$ and $v : E(\base)$ happen to be definitionally equal, we will not be able to map $\base$ to both of them, as required by the computation rules. 

This poses more than just a conceptual problem - in mathematics, we often have several possible definitions of a given notion, all of which are interchangeable from the point of view of a ``user". Having two definitions of a circle which are not (known to be) interchangeable, however, can be problematic: any theorem we establish about or by appealing to $\Sna{1}$ might no longer hold (or even type-check!) when using $\Sn{1}$ instead.

This provides some motivation for considering inductive types with \emph{propositional computation rules} instead. In the case of $\Sn{1}$, the propositional equality at the 0-cell level is witnessed by a path $\beta_{E,e,d}$:
\[ \beta_{E,e,d} : \id[E(\base)]{\circind{E}{e}{d}(\base)}{e} \]
The computation rule at the 1-cell level now states that we have the following commuting diagram:
\begin{center}
\begin{tikzpicture}
\node (N0) at (3,1) {=};
\node (N1) at (0,2) {$\trans{E}{\lp}(\circind{E}{e}{d}(\base))$};
\node (N2) at (6,2) {$\circind{E}{e}{d}(\base)$};
\node (N3) at (0,0) {$\trans{E}{\lp}(e)$};
\node (N4) at (6,0) {$e$};
\draw[-] (N1) -- node[above]{\footnotesize $\dap{\circind{E}{e}{d}}{\lp}$} (N2);
\draw[-] (N1) -- node[left]{\footnotesize $\ap{\trans{E}{\lp}}{\beta_{E,e,d}}$} (N3);
\draw[-] (N2) -- node[right]{\footnotesize $\beta_{E,e,d}$} (N4);
\draw[-] (N3) -- node[below]{\footnotesize $d$} (N4);
\end{tikzpicture}
\end{center}
Similarly, in the case of $\Sna{1}$ we have paths $\gamma_{E,u,v,\mu,\nu}$ and $\delta_{E,u,v,\mu,\nu}$ witnessing the 0-cell propositional equalities:
\[ \gamma_{E,u,v,\mu,\nu} : \id[E(\north)]{\circaind{E}{u}{v}{\mu}{\nu}(\north)}{u} \]
\[ \delta_{E,u,v,\mu,\nu} : \id[E(\south)]{\circaind{E}{u}{v}{\mu}{\nu}(\south)}{v} \]
The computation rule for the constructor $\e$ takes the form of the following commuting diagram:
\begin{center}
\begin{tikzpicture}
\node (N0) at (3.5,1) {=};
\node (N1) at (0,2) {$\trans{E}{\e}(\circaind{E}{u}{v}{\mu}{\nu}(\north))$};
\node (N2) at (7,2) {$\circaind{E}{u}{v}{\mu}{\nu}(\south)$};
\node (N3) at (0,0) {$\trans{E}{\e}(u)$};
\node (N4) at (7,0) {$v$};
\draw[-] (N1) -- node[above]{\footnotesize $\dap{\circaind{E}{u}{v}{\mu}{\nu}}{\e}$} (N2);
\draw[-] (N1) -- node[left]{\footnotesize $\ap{\trans{E}{\e}}{\gamma_{E,u,v,\mu,\nu}}$} (N3);
\draw[-] (N2) -- node[right]{\footnotesize $\delta_{E,u,v,\mu,\nu}$} (N4);
\draw[-] (N3) -- node[below]{\footnotesize $\mu$} (N4);
\end{tikzpicture}
\end{center}
There is an analogous commuting diagram for the constructor $\w$.

It is not too hard to show that Prop.~\ref{thm_circeq} still holds when the computation laws are propositional:
\begin{proposition}
In the setting of inductive types with propositional computation laws, we have $\Sn{1} \simeq \Sna{1}$.
\end{proposition}
At this point it is convenient to establish some terminology.

\subsection{Algebras}
Given a type $C : \U_i$ with a point $c:C$ and path $s : \id[C]{c}{c}$, we can pack the type together with all the operators into a single structure called an
$\Sn{1}$-algebra; we can similarly define an $\Sna{1}$-algebra:
\begin{definition}
We define the type of $\Sn{1}$-algebras on a universe $\U_i$ as \[ \circalg{\U_i} \defeq \sm{C:\U_i} \sm{c:C} (\id{c}{c}) \]
\end{definition}
\begin{definition}
We define the type of $\Sna{1}$-algebras on a universe $\U_i$ as \[ \circaalg{\U_i} \defeq \sm{C:\U_i} \sm{a,b:C} (\id{a}{b}) \times (\id{a}{b}) \]
\end{definition}
\begin{proposition}
We have maps \[\circalgtocircaalg : \circalg{\U_i} \to \circaalg{\U_i}\] \[\circaalgtocircalg : \circaalg{\U_i} \to \circalg{\U_i}\] which form a quasi-equivalence; thus $\circalg{\U_i} \simeq \circaalg{\U_i}$.
\end{proposition}
\begin{proof}
Define the maps by 
\begin{align*}
(C,c,s) & \mapsto (C,c,c,s,\refl{C}{c}) \\
(C,a,b,p,q) & \mapsto (C,a,p \ct \opp{q})
\end{align*}
\end{proof}
For any such algebra, the satisfaction of the principle of dependent elimination into a universe $\U_j$ is now a property internal to the type theory:
\begin{notation}
Define a predicate on the type $\circalg{\U_i}$ by
\begin{align*}
\hascircind{\U_j}(C,c,s) \defeq \prd{E:C \to \U_j}\prd{e:E(c)}\prd{d:\id{\trans{E}{s}(e)}{e}} 
\sm{f:\prd{x}E(x)}\sm{\beta: \id{f(c)}{e}} (\id{\dap{f}{s} \ct \beta}{\ap{\trans{E}{s}}{\beta}} \ct d)
\end{align*}
\end{notation}
\begin{notation}
Define a predicate on the type $\circaalg{\U_i}$ by
\begin{align*}
\hascircaind{\U_j}(C,a,b,p,q) \defeq \prd{E:C \to \U_j} \prd{u:E(a)} \prd{v:E(b)} \prd{\mu:\id{\trans{E}{p}(u)}{v}} \prd{\nu:\id{\trans{E}{{q}}(u)}{v}} \\
\sm{f:\prd{x}E(x)} \sm{\gamma: \id{f(a)}{u}} \sm{\delta: \id{f(b)}{v}}
(\id{\dap{f}{p} \ct \delta}{\ap{\trans{E}{{p}}}{\gamma}} \ct \mu) \; \times
(\id{\dap{f}{q} \ct \delta}{\ap{\trans{E}{q}}{\gamma}} \ct \nu)
\end{align*}
\end{notation}
We can now show that the two induction principles are indeed equivalent:
\begin{proposition}
For any $\mathcal{X} : \circalg{\U_i}$ and $\mathcal{Y} : \circaalg{\U_i}$,
\begin{align*} 
& \hascircind{\U_j}(\mathcal{X}) \to \hascircaind{\U_j}(\circalgtocircaalg(\mathcal{X})) \\ 
& \hascircaind{\U_j}(\mathcal{Y}) \to \hascircind{\U_j}(\circaalgtocircalg(\mathcal{Y}))
\end{align*}
\end{proposition}
\begin{corollary}
The type $\Sn{1}$ satisfies the induction and propositional computation laws for $\Sna{1}$, with $\base$, $\base$, $\lp$, $\refl{\Sn{1}}{\base}$ acting as the constructors $\north$, $\south$, $\e$, $\w$ respectively.
\end{corollary}
\begin{corollary}
The type $\Sna{1}$ satisfies the induction and propositional computation laws for $\Sn{1}$, with $\north$, $\e \ct \opp{\w}$ acting as the constructors $\base$, $\lp$ respectively.
\end{corollary}
Finally, we point out that if a given algebra admits the principle of dependent elimination into a universe $\U_j$, it does so in a unique way:
\begin{proposition}
The types $\hascircind{\U_j}(\mathcal{X})$ and $\hascircaind{\U_j}(\mathcal{Y})$ are mere propositions for any $\mathcal{X} : \circalg{\U_i}$ and $\mathcal{Y} : \circaalg{\U_i}$.
\end{proposition}
This follows from Props.~\ref{thm_circ_enc},~\ref{thm_circa_enc}, and Thm.~\ref{thm_main} in Section~\ref{algebras}.

\subsection{Propositional truncation}
Another example of a higher inductive type is the \emph{propositional truncation} $\trunc{A}:\U_i$ of a type $A:\U_i$, investigated in \cite{awodey_bauer} in an extensional setting under the name of bracket types. Intuitively, $\trunc{A}$ represents the ``squashing" of $A$ which makes all the elements in $A$ equal. The need for propositional truncation arises when we wish to hide information: we want to indicate that $A$ is inhabited without having to give the actual witness $a : A$. For instance, let $P : A \to \U_j$ be a family of a mere propositions. Having a $b : \sm{x:A} P(x)$ is very different from having a $b : \trunc{\sm{x:A} P(x)}$; in the former case, we can directly construct a point in $A$ for which $P$ holds, namely $\fst{b}$. In the latter case, we only know $P$ must hold for \emph{some} point in $A$ but we do not have a generic way of accessing it.

Specifically, we define $\trunc{A}$ as the higher inductive type generated by a constructor $\inj{\cdot}$, which projects a given element of $A$ down to $\trunc{A}$, and a truncation constructor, which states that $\trunc{A}$ is indeed a mere proposition\footnote{Hence the name \emph{propositional} truncation; see Chpt.~6 of \cite{hott} for other kinds of truncation.}:
\begin{align*}
\inj{\cdot} & : A \to \trunc{A} \\
\squash &: \prd{x,y:\trunc{A}} (\id[\trunc{A}]{x}{y})
\end{align*}

As usual, the recursion principle states that given a structure of the same form, we have a function out of $\trunc{A}$ which preserves the constructors:
\begin{mathpar}
\inferrule{C : \U_j \\ c : A \to C \\ s : \prd{x,y:C} (\id[C]{x}{y})}{\truncrec{C}{c}{s} : \trunc{A} \to C}
\end{mathpar}
where for each $a : A$ and $k,l : \trunc{A}$ we have 
\begin{align*}
\truncrec{C}{c}{s}(\inj{a}) & \equiv c(a) : C \\
\id{\ap{\truncrec{C}{c}{s}}{\squash(k,l)} &}{s(\truncrec{C}{c}{s}(k), \truncrec{C}{c}{s}(l))}
\end{align*}
We note that we are only able to eliminate into types which are themselves mere propositions. This together with Prop.~\ref{thm_contr_path} implies that the second computation law always holds.

To state the induction principle, we need to suitably generalize the last hypothesis. As before, we note that once the desired map $f : \prd{x:\trunc{A}} E(x)$ is constructed, it will give us a path from $\trans{E}{\squash(k,l)}(f(k))$ to $f(l)$ in $E(l)$ for any $k,l : \trunc{A}$. Hence, $E$ should already come equipped with such a family of paths - except, of course, we have no way of referring to $f(k)$ and $f(l)$ before $f$ is constructed. Thus, we simply require that such a path exists for \emph{any} points $u : E(k)$ and $v : E(l)$:
\begin{mathpar}
\inferrule{E : \trunc{A} \to \U_j \\ e : \prd{a:A} E(\inj{a}) \\ d : \prd{x,y:\trunc{A}}\prd{u:E(x)}\prd{v:E(y)}(\id[E(y)]{\trans{E}{\squash(x,y)}(u)}{v})}{\truncind{E}{e}{d} : \prd{x:\trunc{A}} E(x)}
\end{mathpar}
For each $a : A$ and $k,l : \trunc{A}$, we have the computation rules
\begin{align*}
\truncind{E}{e}{d}(\inj{a}) & \equiv e(a) : E(\inj{a}) \\
\id{\dap{\truncind{E}{e}{d}}{\squash(k,l)} &}{d(k,l,\truncind{E}{e}{d}(k), \truncind{E}{e}{d}(l))}
\end{align*}

The second rule again turns out to always hold, as we will see shortly; however, we first note that this definition of $\trunc{A}$ has its share of problems. For instance, as the type $\nat$ of natural numbers is inhabited, it follows that $\trunc{\nat} = \one$. It is not obvious, however, how to turn $\one$ itself into a truncation of $\nat$, since the first computation law ought to hold \emph{definitionally}. More disturbing yet is the observation by N. Kraus in \cite{trunc_inverse} that there exists a map $f$ such that $f \circ \inj{\cdot} \equiv \idfun{\nat}$; this is another surprising side effect of definitional computation law for $\inj{\cdot}$.

Both issues can be avoided by using propositional computation rules instead, i.e., we have
\begin{align*}
\prd{a:A}(\id[C]{\truncrec{C}{c}{s}(\inj{a})&}{c(a)}) \\
\prd{a:A}(\id[E(\inj{a})]{\truncind{E}{e}{d}(\inj{a}) &}{e(a)})
\end{align*}

\subsection{Algebra homomorphisms}
We can again pack all the operators into a single structure:
\begin{definition}
Define the type of $\trunc{A}$-algebras on a universe $\U_j$ as \[\truncalg{\U_j} \defeq \sm{C:\U_j} (A \to C) \times \isprop{A}\]
\end{definition}
We can also talk about homomorphisms between two $\trunc{A}$-algebras, which are mappings that preserve all operators:
\begin{definition}
For  $\mathcal{X} : \truncalg{\U_j}$ and $\mathcal{Y} : \truncalg{\U_k}$, define the type of homomorphisms from $\mathcal{X}$ to $\mathcal{Y}$ by
\begin{align*}
& \trunchom \; (C,c,s) \; (D,d,r) \defeq \sm{f:C \to D} (\prd{a:A} (f(c \; a) = d \; a)) \times \prd{k,l : C} (\ap{f}{s \; k \; l} = r(f \; k,f \; l))
\end{align*}
\end{definition}
The recursion principle into $\U_k$ can thus be expressed as:
\begin{notation}
For $\mathcal{X} : \truncalg{\U_j}$, define \[\hastruncrec{\U_k}(\mathcal{X}) \defeq \prd{\mathcal{Y} : \truncalg{\U_k}} \trunchom \; \mathcal{X} \; \mathcal{Y}\]
\end{notation}
We also have a dependent version of these concepts:
\begin{definition}
Define the type of fibered $\trunc{A}$-algebras on a universe $\U_k$ over $\mathcal{X} : \truncalg{\U_j}$ by
\begin{align*}
\truncfibalg{\U_k} \; (C,c,s) \defeq \sm{E:C \to \U_k} (\prd{x:A} E(c \; x)) \times \prd{x,y:C}\prd{u:E(x)}\prd{v:E(y)}(\id{\trans{E}{s(x,y)}(u)}{v})
\end{align*}
\end{definition}
\begin{definition}
For  $\mathcal{X} : \truncalg{\U_j}$ and $\mathcal{Y} : \truncfibalg{\U_k} \; \mathcal{X}$, define the type of fibered homomorphisms from $\mathcal{X}$ to $\mathcal{Y}$ by
\begin{align*}
& \truncfibhom \; (C,c,s) \; (E,e,d) \defeq \sm{f:\prd{x:C} E(x)} \\ & \;\;\; (\prd{a:A} (f(c \; a) = e \; a)) \times \prd{k,l:C} (\dap{f}{s \; k \; l} = d(f \; k,f \; l))
\end{align*}
\end{definition}
The induction principle into $\U_k$ can thus be expressed as:
\begin{notation}
For $\mathcal{X} : \truncalg{\U_j}$, define \[\hastruncind{\U_k}(\mathcal{X}) \defeq \prd{(\mathcal{Y} : \truncfibalg{\U_k} \; \mathcal{X})} \truncfibhom \; \mathcal{X} \; \mathcal{Y}\]
\end{notation}
We now observe that the complicated last expression in the definition of fibered algebras can be replaced by saying that $E$ is a family of mere propositions:
\begin{proposition}\label{thm_fibalg}
Define the type of fibered $\trunc{A}$-algebras on a universe $\U_k$ over $\mathcal{X} : \truncalg{\U_j}$ alternatively by
\begin{align*}
\truncfibalg{\U_k}' \; (C,c,s) \defeq \sm{E:C \to \U_k} (\prd{x:A} E(c \; x)) \times \prd{x:C} \isprop{E(x)}
\end{align*}
Then for any $\mathcal{X}$, $\truncfibalg{\U_k} \; \mathcal{X} \simeq \truncfibalg{\U_k}' \; \mathcal{X}$.
\end{proposition}
By Props.~\ref{thm_contr_path} and~\ref{thm_fibalg}, the induction and recursion principles now take a particularly simple form, as implied by:
\begin{proposition}\label{thm_truncequiv}
Given any algebras $\mathcal{X} : \truncalg{\U_j}$, $\mathcal{Y} : \truncalg{\U_k}$, $\mathcal{Z} : \truncfibalg{\U_k} \; \mathcal{X}$,
we have 
\begin{align*}
\trunchom \; \mathcal{X} \; \mathcal{Y} & \simeq \fst{\mathcal{X}} \to \fst{\mathcal{Y}} \\
\truncfibhom \; \mathcal{X} \; \mathcal{Z} & \simeq \prd{x:\fst{\mathcal{X}}} \fst{\mathcal{Z}}(x)
\end{align*}
\end{proposition}
\begin{corollary}
For any $\mathcal{X} : \truncalg{\U_j}$ and $\U_k$, the types $\hastruncrec{\U_k}(\mathcal{X})$ and $\hastruncind{\U_k}(\mathcal{X})$ are mere propositions.
\end{corollary}
Finally, we can show that induction and recursion for $\trunc{A}$ are in fact equivalent. We note that since universe levels are cumulative, the technical restriction that $k \geq j$ does not pose a problem.
\begin{proposition}\label{thm_trunc_main}
For any $\mathcal{X} : \truncalg{\U_j}$ and $\U_k$ with $k \geq j$, we have \[\hastruncrec{\U_k}(\mathcal{X}) \simeq \hastruncind{\U_k}(\mathcal{X})\]
\end{proposition}
\begin{proof}
The direction from right to left is obvious. For the other direction, let the algebras $(C,c,s) : \truncalg{\U_j}$ and $(E,e,d) : \truncfibalg{\U_k} \; (C,c,s)$ be given. The total space $\sm{x:C} E(c) : \U_k$ is a mere proposition, we can thus apply recursion with the term $\lam{x:A} (c \; x,e \; x)$ to get a function $u : \trunc{A} \to \sm{x:C} E(c)$. We have a homotopy $\alpha : \mathsf{fst} \circ u \sim \idfun{C}$ as $C$ is a mere proposition. 
Applying second projection and transporting gives us a function $\lam{x:\trunc{A}} \trans{E}{\alpha(x)}(\snd{u \; x})$.
\end{proof}

\section{Homotopy-Initial Algebras}\label{algebras}
Here we develop an equivalent characterization of higher inductive types as \emph{homotopy-initial} (``\emph{h-initial}") algebras \cite{wtypes}. We will work with a slightly larger class of HITs, which we call \emph{$\W$-suspensions}. Informally, a $\W$-suspension is generated by any number of \emph{points} and any number of paths (\emph{cells}) between any two generating points. Formally, given types $B,C : \U_i$, a type family $A : B \to \U_i$, and functions $f,g : B \to C$, the $\W$-suspension $\wsusp{A}{f}{g} : \U_i$ is the higher inductive type generated by the constructors
\begin{align*}
& \point : C \to \wsusp{A}{f}{g} \\
& \cell : \prd{b:B} A(b) \to \id[\wsusp{A}{f}{g}]{\point(f\;b)}{\point(g\;b)}
\end{align*}
Thus, $C$ can be thought of as the index type for points, $B$ as the index type for the different endpoint configurations, $f$ and $g$ as determining the start- and endpoints of a particular configuration, and $A(b)$ as the index type for the different paths between the two points specified by $b$.

We can encode the circle $\Sn{1}$ by taking $C,B \defeq \one$, $A \defeq \lam{\_:\one}\one$, and $f,g \defeq \lam{\_:\one} \star$. The circle $\Sna{1}$ arises when we take $C \defeq \two$, $B \defeq \one$, $A \defeq \lam{\_:\one}\two$, $f \defeq \lam{\_:\one} \top$, $g \defeq \lam{\_:\one} \bot$. Other types which can be represented in this form include the interval type and suspensions (Chpt.~6 of \cite{hott}), hence in particular all the higher spheres $\mathbb{S}^n$.

We have the expected recursion principle:
\begin{mathpar}
\inferrule{X : \U_j \\ p : C \to X \\ s : \prd{b:B} A(b) \to \id[X]{p(f\;b)}{p(g\;b)}}{\wsusprec{X}{p}{s} : \wsusp{A}{f}{g} \to X}
\end{mathpar}
with the computation laws
\[ \beta^{A,f,g}_{X,p,s} : \prd{c:C} (\wsusprec{X}{p}{s}(\point(c)) = p(c)) \]
and
\begin{center}
\begin{tikzpicture}
\node (N0) at (3.5,1) {=};
\node (N1) at (0,2) {$\wsusprec{X}{p}{s}(\point(f\;b))$};
\node (N2) at (7,2) {$\wsusprec{X}{p}{s}(\point(g\;b))$};
\node (N3) at (0,0) {$p(f\;b)$};
\node (N4) at (7,0) {$p(g\;b)$};
\draw[-] (N1) -- node[above]{\footnotesize $\ap{\wsusprec{X}{p}{s}}{\cell(b,a)}$} (N2);
\draw[-] (N1) -- node[left]{\footnotesize $\beta^{A,f,g}_{X,p,s}(f\;b)$} (N3);
\draw[-] (N2) -- node[right]{\footnotesize $\beta^{A,f,g}_{X,p,s}(g\;b)$} (N4);
\draw[-] (N3) -- node[below]{\footnotesize $s(b,a)$} (N4);
\end{tikzpicture}
\end{center}
for each $b:B$, $a : A(b)$. Similarly, we have the induction principle
\begin{mathpar}
\inferrule{E : \wsusp{A}{f}{g} \to \U_j \\ e : \prd{c:C} E(\point(c)) \\ d : \prd{b:B} \prd{a:A(b)} (\id[E(\point(g \; b))]{\trans{E}{\cell(b,a)}(e(f\;b))}{e(g\;b)})}{\wsuspind{E}{e}{d} : \prd{x:\wsusp{A}{f}{g}} E(x)}
\end{mathpar}
with the computation laws
\[ \beta^{A,f,g}_{E,e,d} : \prd{c:C} (\wsuspind{E}{e}{d}(\point(c)) = p(c)) \]
and
\begin{center}
\begin{tikzpicture}
\node (N0) at (3.5,1) {=};
\node (N1) at (0,2) {$\trans{E}{\cell(b,a)}(\wsuspind{E}{e}{d}(\point(f\;b)))$};
\node (N2) at (7,2) {$\wsuspind{E}{e}{d}(\point(g\;b))$};
\node (N3) at (0,0) {$\trans{E}{\cell(b,a)}(p(f\;b))$};
\node (N4) at (7,0) {$p(g\;b)$};
\draw[-] (N1) -- node[above]{\footnotesize $\dap{\wsuspind{E}{e}{d}}{\cell(b,a)}$} (N2);
\draw[-] (N1) -- node[left]{\footnotesize $\ap{\trans{E}{\cell(b,a)}}{\beta^{A,f,g}_{E,e,d}(f\;b)}$} (N3);
\draw[-] (N2) -- node[right]{\footnotesize $\beta^{A,f,g}_{E,e,d}(g\;b)$} (N4);
\draw[-] (N3) -- node[below]{\footnotesize $s(b,a)$} (N4);
\end{tikzpicture}
\end{center}
for each $b:B$, $a:B(a)$.

Following the now-familiar pattern, we define W-suspension algebras and homomorphisms:
\begin{definition}
We define the type of $\wsusp{A}{f}{g}$-algebras on a universe $\U_j$ to be 
\begin{align*}
\wsuspalg{\U_j} \defeq \sm{X:\U_j} \sm{p:C \to X} \prd{b:B} A(b) \to \id{p(f\;b)}{p(g\;b)}
\end{align*}
\end{definition}
\begin{definition}
Define the type of fibered $\wsusp{A}{f}{g}$-algebras on a universe $\U_k$ over $\mathcal{X} : \wsuspalg{\U_j}$ by
\begin{align*}
\wsuspfibalg{\U_k} \; (X,p,s) \defeq \sm{E:X \to \U_k} \sm{(e: \prd{c:C} E(p \; c))} \prd{b:B} \prd{a:A(b)} (\id{\trans{E}{s(b,a)}(e(f\;b))}{e(g\;b)})
\end{align*}
\end{definition}
\begin{definition}
For  $\mathcal{X} : \wsuspalg{\U_j}$ and $\mathcal{Y} : \wsuspalg{\U_k}$, define the type of homomorphisms from $\mathcal{X}$ to $\mathcal{Y}$ by
\begin{align*}
& \wsusphom \; (X,p,s) \; (Y,q,r) \defeq \sm{h:X \to Y} \sm{\beta:\prd{c:C} (\id{h(p(c))}{q(c)})} \\ & \;\; \prd{b:B}\prd{a:A} (\ap{h}{s(b,a)} \ct \beta(g\;b) = \beta(f\;b) \ct r(b,a))
\end{align*}
\end{definition}
\begin{definition}\label{def_fibhom}
For  $\mathcal{X} : \wsuspalg{\U_j}$ and $\mathcal{Y} : \wsuspfibalg{\U_k} \; \mathcal{X}$, define the type of fibered homomorphisms from $\mathcal{X}$ to $\mathcal{Y}$ by
\begin{align*}
& \wsuspfibhom \; (X,p,s) \; (E,e,d) \defeq \sm{(h:\prd{x:X} E(x))} \sm{(\beta:\prd{c:C} (\id{h(p(c))}{e(c)}))} \\ & \;\;\; \prd{b:B}\prd{a:A} (\dap{h}{s(b,a)} \ct \beta(g\;b) = \ap{\trans{E}{s(b,a)}}{\beta(f\;b)} \ct d(b,a))
\end{align*}
\end{definition}
\begin{notation}
For $\mathcal{X} : \wsuspalg{\U_j}$, define
\begin{align*}
& \haswsusprec{\U_k}(\mathcal{X}) \defeq \prd{\mathcal{Y} : \wsuspalg{\U_k}} \wsusphom \; \mathcal{X} \; \mathcal{Y}\\
& \haswsuspind{\U_k}(\mathcal{X}) \defeq \prd{(\mathcal{Y} : \wsuspfibalg{\U_k} \; \mathcal{X})} \wsuspfibhom \; \mathcal{X} \; \mathcal{Y}
\end{align*}
\end{notation}

We can now show that our encodings of the circles $\Sn{1}$ and $\Sna{1}$ as $\W$-suspensions are indeed correct:
\begin{proposition}\label{thm_circ_enc}
Let $A \defeq \lam{\_:\one}\one$ and $f,g \defeq \lam{\_:\one} \star$. There are functions $\circalgtowsuspalg$ and $\wsuspalgtocircalg$ between $\circalg{\U_i}$ and $\wsuspalg{\U_i}$ which comprise a quasi-equivalence. Also, 
\begin{align*}
\hascircind{\U_k}(\mathcal{X}) & \simeq \haswsuspind{\U_k}(\circalgtowsuspalg(\mathcal{X})) \\
\haswsuspind{\U_k}(\mathcal{Y}) & \simeq \hascircind{\U_k}(\wsuspalgtocircalg(\mathcal{Y}))
\end{align*}
\end{proposition}
\begin{proof}[Proof sketch]
Define the maps by 
\begin{align*}
(C,c,s) & \mapsto (C,\lam{\_:\one} c,\lam{\_:\one} \lam{\_:\one} s) \\
(C,p,s) & \mapsto (C,p(\star),s(\star\,\star))
\end{align*}
\end{proof}
\begin{proposition}\label{thm_circa_enc}
Let $A \defeq \lam{\_:\one}\two$, $f \defeq \lam{\_:\one} \top$, $g \defeq \lam{\_:\one} \bot$. There are maps $\circaalgtowsuspalg$ and $\wsuspalgtocircaalg$ between $\circaalg{\U_i}$ and $\wsuspalg{\U_i}$ which comprise a quasi-equivalence. Also, 
\begin{align*}
\hascircaind{\U_k}(\mathcal{X}) & \simeq \haswsuspind{\U_k}(\circaalgtowsuspalg(\mathcal{X})) \\
\haswsuspind{\U_k}(\mathcal{Y}) & \simeq \hascircaind{\U_k}(\wsuspalgtocircaalg(\mathcal{Y}))
\end{align*}
\end{proposition}
\begin{proof}[Proof sketch]
Define the maps by 
\begin{align*}
(C,a,b,p,q) & \mapsto (C,\recsym^\two_{C,a,b},\lam{\_:\one} \recsym^\two_{\recsym^\two_{C,a,b}(\top) \to \recsym^\two_{C,a,b}(\bot),\; \gamma\ct p \ct \delta^{-1},\; \gamma \ct q \ct \delta^{-1}}) \\
(C,p,s) & \mapsto (C,p(\top),p(\bot),s(\star,\top),s(\star,\bot))
\end{align*}
where $\gamma : \recsym^\two_{C,a,b}(\top) = a$ and $\delta : \recsym^\two_{C,a,b}(\bot) = b$ witness the two computation rules for $\recsym^\two$. 
\end{proof}

\subsection{Main theorem}
First we define the universal property of homotopy-initiality \cite{wtypes}, which translates the notion of uniqueness into the homotopical setting as contractibility:
\begin{definition}
We call an algebra $\mathcal{X}:\wsuspalg{\U_j}$ \emph{homotopy-initial} on the universe $\U_k$ if the space of homomorphisms from $\mathcal{X}$ to any other algebra on $\U_k$ is contractible:
\[ \hinitial{\U_k}(\mathcal{X}) \defeq \prd{\mathcal{Y}:\wsuspalg{\U_k}} \iscontr{\wsusphom \; \mathcal{X} \; \mathcal{Y}}\]
\end{definition}

We now want to show that homotopy-initiality is in fact equivalent to the induction principle. As an intermediate step, we show that the induction principle can be reduced to the recursion principle plus a certain uniqueness condition, which we call $\hasuniq{\U_k}( \mathcal{X})$. 

A uniqueness condition is needed since in general, the recursion principle does not fully determine an inductive type: the recursion principle for the circle, for example, is also satisfied by the disjoint union of \emph{two} circles. Additionally, by Cor.~\ref{thm_contr_char} we see that the property of $\mathcal{X}$ being h-initial means that for any $\mathcal{Y}$, there exists a homomorphism from $\mathcal{X}$ to $\mathcal{Y}$, and furthermore, that any two such homomorphisms are equal. The existence assertion is precisely the recursion principle; the equality assertion turns out to be equivalent to our uniqueness condition, which is presented in a more explicit form.

To phrase the uniqueness condition in a compact way, we introduce some more notation:
\begin{notation}
Given $p : C \to X$, $e : \prd{c:C} E(x)$, define the type of pointed functions as
\[\poinfun{C}{E}{p}{e} \defeq \sm{(h : \prd{x:X} E(x))} \prd{c:C} (h(p \; x) = e(x)) \]
Define the type of homotopies between two pointed functions $\theta,\phi : \poinfun{C}{E}{p}{e}$ by
\[ (h,\gamma) \sim (i,\delta) \defeq \sm{\alpha : h \sim i} \prd{c:C} (\gamma(c) = \alpha(p\;c) \ct \delta(c)) \]
\end{notation}
Of course, any (fibered) homomorphisms $\mu,\nu$ determine pointed functions; denote by $\mu \sim \nu$ the type of homotopies between these pointed functions. 
\begin{notation}\label{nat_coh}
Given $\mathcal{X} : \wsuspalg{\U_j}$, $\mathcal{Y} : \wsuspalg{\U_k}$, $\mu,\nu : \wsusphom \; \mathcal{X} \; \mathcal{Y}$, and $b:B$, $a:A(b)$, we define a type family on $\mu \sim \nu$ by mapping
\[ \wsuspcoh \; b\; a\; (X,p,s) \; (Y,q,r) \; (h,\gamma,\Theta) \; (i,\delta,\Phi) \; (\alpha, \eta)\]
to the type asserting the commutativity of the following diagram:

\begin{center}
\begin{tikzpicture}
\node (N0) at (3.5,3) {=};
\node (N1) at (0,6) {$\alpha_{p(f \; b)} \ct \ap{i}{s(b,a)}$};
\node (N2) at (7,6) {$\ap{h}{s(b,a)} \ct \alpha_{p(g\;b)}$};
\node (N3) at (0,4) {$(\gamma_{f(b)} \ct \delta^{-1}_{f(b)}) \ct \ap{i}{s(b,a)}$};
\node (N4) at (7,4) {$\ap{h}{s(b,a)} \ct (\gamma_{g(b)} \ct \delta^{-1}_{g(b)})$};
\node (N5) at (0,2) {$\gamma_{f(b)} \ct (\delta^{-1}_{f(b)} \ct \ap{i}{s(b,a)})$};
\node (N6) at (7,2) {$(\ap{h}{s(a,b)} \ct \gamma_{g(b)}) \ct \delta^{-1}_{g(b)}$};
\node (N7) at (0,0) {$\gamma_{f(b)} \ct (r(b,a) \ct \delta^{-1}_{g(b)})$};
\node (N8) at (7,0) {$(\gamma_{f(b)} \ct r(b,a)) \ct \delta^{-1}_{g(b)}$};
\draw[-] (N1) -- node[above]{\scriptsize \emph{naturality of} $\alpha$} (N2);
\draw[-] (N1) -- node[left]{\scriptsize \emph{via} $\invtri(\eta_{f(b)})$} (N3);
\draw[-] (N2) -- node[right]{\scriptsize \emph{via} $\invtri(\eta_{g(b)})$} (N4);
\draw[-] (N3) -- node[above]{\footnotesize} (N5);
\draw[-] (N4) -- node[above]{\footnotesize} (N6);
\draw[-] (N5) -- node[left]{\scriptsize \emph{via} $\invsq(\Phi(a,b))$} (N7);
\draw[-] (N6) -- node[right]{\scriptsize \emph{via} $\Theta(a,b)$} (N8);
\draw[-] (N7) -- node[above]{\footnotesize} (N8);
\end{tikzpicture}
\end{center}
\end{notation}
For brevity, we will usually leave out some of the arguments to $\wsuspcoh$ as appropriate. The maps $\invtri : (u = v \ct w) \to (v = u \ct \opp{w})$ and $\invsq : (u \ct v = w \ct z) \to (\opp{w} \ct u = z \ct \opp{v})$ perform the obvious manipulations of diagrams. 
It is useful to fix a specific definition: we define $\invtri$ by path induction on $w$ so that for $\epsilon : u =_{x=y} (v \ct \mathsf{refl}(y))$, the path $\invtri(\epsilon)$ is definitionally equal to
\begin{center}
\begin{tikzpicture}
\node (N0) at (0,0) {$v$};
\node (N1) at (2,0) {$v \ct \mathsf{refl}(y)$};
\node (N2) at (4,0) {$u$};
\node (N3) at (6,0) {$u \ct \mathsf{refl}(y)$};
\draw[-] (N0) -- node[above]{} (N1);
\draw[-] (N1) -- node[above]{\scriptsize via $\epsilon$} (N2);
\draw[-] (N2) -- node[above]{} (N3);
\end{tikzpicture}
\end{center}
Similarly, we define $\invsq$ by path induction on $v$ and $w$ so that for $\epsilon : (u \ct \mathsf{refl}(y)) =_{x=y} (\mathsf{refl}(x) \ct z)$, the path $\invsq(\epsilon)$ is definitionally equal to
\begin{center}
\begin{tikzpicture}
\node (N0) at (0,0) {$ \mathsf{refl}(x) \ct u$};
\node (N1) at (2,0) {$u$};
\node (N2) at (4,0) {$u \ct \mathsf{refl}(y)$};
\node (N3) at (7,0) {$\mathsf{refl}(x) \ct z$};
\node (N4) at (9,0) {$z$};
\node (N5) at (11,0) {$z \ct \mathsf{refl}(y)$};
\draw[-] (N0) -- node[above]{} (N1);
\draw[-] (N1) -- node[above]{} (N2);
\draw[-] (N2) -- node[above]{\scriptsize $\epsilon$} (N3);
\draw[-] (N3) -- node[above]{} (N4);
\draw[-] (N4) -- node[above]{} (N5);
\end{tikzpicture}
\end{center}
There are maps $\invtri^{-1} :  (v = u \ct \opp{w}) \to (u = v \ct w)$ and $\invsq^{-1} : (\opp{w} \ct u = z \ct \opp{v}) \to (u \ct v = w \ct z)$ which form quasi-equivalences with $\invtri$ and $\invsq$ respectively.
 
\begin{definition}
Given $\mathcal{X} : \wsuspalg{\U_j}$, $\mathcal{Y} : \wsuspalg{\U_k}$ and $\mu,\nu : \wsusphom \; \mathcal{X} \; \mathcal{Y}$, define the type of algebra 2-cells between $\mu$ and $\nu$ as
\[ \twocell \; \mu \; \nu = \sm{\mathfrak{p} : \mu \sim \nu} \prd{b:B} \prd{a:A(b)} \wsuspcoh \; b \; a  \; \mathfrak{p}\]
\end{definition}
Our uniqueness condition then says that for any algebra $\mathcal{Y}$ and homomorphisms $\mu,\nu$ from $\mathcal{X}$ to $\mathcal{Y}$, there exists an algebra \emph{2-cell} between $\mu$ and $\nu$:
\begin{notation}
For $\mathcal{X} : \wsuspalg{\U_j}$, define
\begin{align*}
\hasuniq{\U_k}(\mathcal{X}) \defeq \prd{\mathcal{Y} : \wsuspalg{\U_k}} \prd{(\mu,\nu : \wsusphom \; \mathcal{X} \; \mathcal{Y})} \twocell \; \mu \; \nu
\end{align*}
\end{notation}
We now come to the main theorem:
\begin{theorem}\label{thm_main}
For any algebra $\mathcal{X} : \wsuspalg{\U_j}$, we have
\begin{align*}
 \haswsuspind{\U_k}(\mathcal{X}) \simeq 
 \haswsusprec{\U_k}(\mathcal{X}) \times \hasuniq{\U_k}(\mathcal{X}) \simeq 
 \hinitial{\U_k}(\mathcal{X})
\end{align*}
for $k \geq j$ and the three types above are mere propositions.
\end{theorem}
We point out that the uniqueness condition was not construed in an ad hoc way; rather, it is systematically derived from the induction principle. We recall that a homomorphism between two algebras $(X,p,s)$, $(Y,q,r)$ is a triple $(h,\beta,\Theta)$, where $h : X \to Y$ is a function between the carrier types, $\beta$ specifies the behavior of $h$ on the \emph{0-cells}, i.e., the value of $h(p(c))$, and $\Theta$ specifies the behavior of $h$ on the \emph{1-cells}, i.e., the value of $\ap{h}{s(b,a)}$. The existence of such a homomorphism for any $(Y,q,r)$ is of course precisely the recursion principle.
Similarly, the uniqueness condition itself can be viewed as a certain form of induction, albeit a very specific one. We recall that an algebra \emph{2-cell} between $(h,\gamma,\Theta)$ and $(i,\delta,\Phi)$ is a triple $(\alpha,\eta,\Psi)$, where $\alpha : h \sim i$ relates the two underlying mappings, $\eta$ relates the path families $\gamma$ and $\delta$, and $\Psi$ relates the proof families $\Theta$ and $\Phi$ with the diagram in Not.~\ref{nat_coh}. The existence of such an algebra \emph{2-cell} between any $(h,\gamma,\Theta)$ and $(i,\delta,\Phi)$ thus guarantees the existence of a dependent function $\alpha : \prd{x:X} (h(x) = i(x))$ - the ``inductor". 
The behavior of $\alpha$ on the \emph{0-cells}, i.e., the value of $\alpha(p(c))$, is specified by the term $\eta$, which thus serves as the first ``computation rule". Finally, the behavior of $\alpha$ on the \emph{1-cells}, i.e., the value of $\dap{\alpha}{s(b,a)}$, is specified by the family of diagrams $\Psi$\footnote{Although the diagram scheme in Not.~\ref{nat_coh} uses an equivalent formulation that
does not explicitly mention the term $\dap{\alpha}{s(b,a)}$.}, which hence serves as the second ``computation rule." 

As a sanity check, we look at the analogue of the main theorem in the case of propositional truncations. By Prop.~\ref{thm_truncequiv}, homomorphisms between $\trunc{A}$-algebras are just maps between the carrier types, and the elimination principles for $\trunc{A}$ do not postulate any computation rules. A \emph{2-cell} between homomorphisms $h$ and $i$ is thus just a homotopy $\alpha : h \sim i$.
The existence of such $\alpha$ is of course a moot point in the setting of mere propositions and the uniqueness condition reduces to the unit type $\one$. Similarly, h-initiality reduces to the recursion principle by virtue of Prop.~\ref{thm_contr_char}. The rest follows from Prop.~\ref{thm_trunc_main}.

Before we proceed to the proof of the main theorem, we consider a dependent version of the uniqueness condition, $\hasinduniq{\U_k}(\mathcal{X})$, which uses the fibered version of algebras, homomorphisms, and algebra \emph{2-cells}:

\begin{notation}\label{nat_coh_dep}
For $\mathcal{X} : \wsuspalg{\U_j}$, $\mathcal{Y} : \wsuspfibalg{\U_k}\;\mathcal{X}$, homomorphisms $\mu,\nu : \wsuspfibhom \; \mathcal{X} \; \mathcal{Y}$ and $b:B$, $a:A(b)$, we define a type family on $\mu \sim \nu$ by mapping
\[ \wsuspfibcoh \; b\; a\; (X,p,s) \; (E,e,d) \; (h,\gamma,\Theta) \; (i,\delta,\Phi) \; (\alpha, \eta)\]
to the type asserting the commutativity of the following diagram:

\begin{center}
\begin{tikzpicture}
\node (N0)  at (4.4,2) {=};

\node (N1)  at (0,6) {$\ap{\trans{E}{s(b,a)}}{\alpha_{p(f \; b)}} \ct \dap{i}{s(b,a)}$};
\node (N1a) at (0,4) {$\ap{\trans{E}{s(b,a)}}{\gamma_{f(b)} \ct \delta^{-1}_{f(b)}} \ct \dap{i}{s(b,a)}$};
\node (N3)  at (0,2) {$(\ap{\trans{E}{s(b,a)}}{\gamma_{f(b)}} \ct \ap{\trans{E}{s(a,b)}}{\delta_{f(b)}}^{-1}) \ct \dap{i}{s(b,a)}\;\;\;\;\;\;\;\;\;\;\;\;\;\;\;\;\;\;\;\;$};
\node (N5)  at (0,0) {$\ap{\trans{E}{s(b,a)}}{\gamma_{f(b)}} \ct (\ap{\trans{E}{s(a,b)}}{\delta_{f(b)}}^{-1} \ct \dap{i}{s(b,a)})\;\;\;\;\;\;\;\;\;\;\;\;\;\;\;\;\;\;\;\;$};
\node (N7)  at (0,-2) {$\ap{\trans{E}{s(b,a)}}{\gamma_{f(b)}} \ct (d(b,a) \ct \delta^{-1}_{g(b)})$};

\node (N2)  at (8,6) {$\dap{h}{s(b,a)} \ct \alpha_{p(g\;b)}$};
\node (N4)  at (8,2) {$\dap{h}{s(b,a)} \ct (\gamma_{g(b)} \ct \delta^{-1}_{g(b)})$};
\node (N6)  at (8,0) {$(\dap{h}{s(a,b)} \ct \gamma_{g(b)}) \ct \delta^{-1}_{g(b)}$};
\node (N8)  at (8,-2) {$(\ap{\trans{E}{s(b,a)}}{\gamma_{f(b)}} \ct d(b,a)) \ct \delta^{-1}_{g(b)}$};
\draw[-] (N1) -- node[above]{\scriptsize \emph{naturality of} $\alpha$} (N2);
\draw[-] (N1) -- node[left]{\scriptsize \emph{via} $\invtri(\eta_{f(b)})$} (N1a);
\draw[-] (N1a) -- node[left]{\scriptsize} (N3);
\draw[-] (N2) -- node[right]{\scriptsize \emph{via} $\invtri(\eta_{g(b)})$} (N4);
\draw[-] (N3) -- node[above]{\footnotesize} (N5);
\draw[-] (N4) -- node[above]{\footnotesize} (N6);
\draw[-] (N5) -- node[left]{\scriptsize \emph{via} $\invsq(\Phi(a,b))$} (N7);
\draw[-] (N6) -- node[right]{\scriptsize \emph{via} $\Theta(a,b)$} (N8);
\draw[-] (N7) -- node[above]{\footnotesize} (N8);
\end{tikzpicture}
\end{center}
\end{notation}

\begin{definition}
Given $\mathcal{X} : \wsuspalg{\U_j}$, $\mathcal{Y} : \wsuspfibalg{\U_k} \; \mathcal{X}$ and $\mu,\nu : \wsuspfibhom \; \mathcal{X} \; \mathcal{Y}$, define the type of algebra 2-cells between $\mu$ and $\nu$ as
\[ \twofibcell \; \mu \; \nu = \sm{\mathfrak{p} : \mu \sim \nu} \prd{b:B} \prd{a:A(b)} \wsuspfibcoh \; b \; a  \; \mathfrak{p}\]
\end{definition}
\begin{notation}
For $\mathcal{X} : \wsuspalg{\U_j}$, define
\begin{align*}
\hasinduniq{\U_k}(\mathcal{X}) \defeq \prd{(\mathcal{Y} : \wsuspalg{\U_k} \; \mathcal{X})} \prd{(\mu,\nu : \wsuspfibhom \; \mathcal{X} \; \mathcal{Y})} \twofibcell \; \mu \; \nu
\end{align*}
\end{notation}

At last, we outline the proof of the main theorem.
\begin{proof}[Proof outline]
The proof consists of the following steps:
\begin{enumerate}
\item[\emph{1)}] \emph{Show that the induction principle implies the recursion principle, see \ref{pf_ind_imp_rec}}.

\item[\emph{2)}] \emph{Show that the induction principle implies both uniqueness conditions, see \ref{pf_ind_imp_coh}}.

\item[\emph{3)}] \emph{Show that the recursion plus uniqueness principles imply the induction principle, see \ref{pf_rec_coh_imp_ind}.}

\item[\emph{4)}] \emph{Show that the space of (fibered) 2-cells between two (fibered) homomorphisms $\mu,\nu$ is equivalent to the path space $\mu = \nu$, see \ref{pf_two_cell_path}}.
\end{enumerate}

The last step together with Prop.~\ref{thm_contr_char} establishes the equivalence of \emph{h-initiality} and \emph{recursion +    
 uniqueness}. Since the former is a mere proposition by Prop.~\ref{thm_contr_prop}, so is the latter. The first three steps establish logical equivalence between \emph{induction} and \emph{recursion + uniqueness}. It remains to prove the former is a mere proposition.

It is sufficient to do so under the assumption that the type $\haswsuspind{\U_k}\;\mathcal{X}$ is inhabited. Thus, the second step tells us that the dependent uniqueness principle holds. By the third step, this means that for any $\mathcal{Y}$, any two fibered homomorphisms from $\mathcal{X}$ to $\mathcal{Y}$ are equal. But of course, this implies that any two inhabitants of $\haswsuspind{\U_k}\;\mathcal{X}$ are equal.
\end{proof}

The relationships between the various properties are depicted in the following diagram:

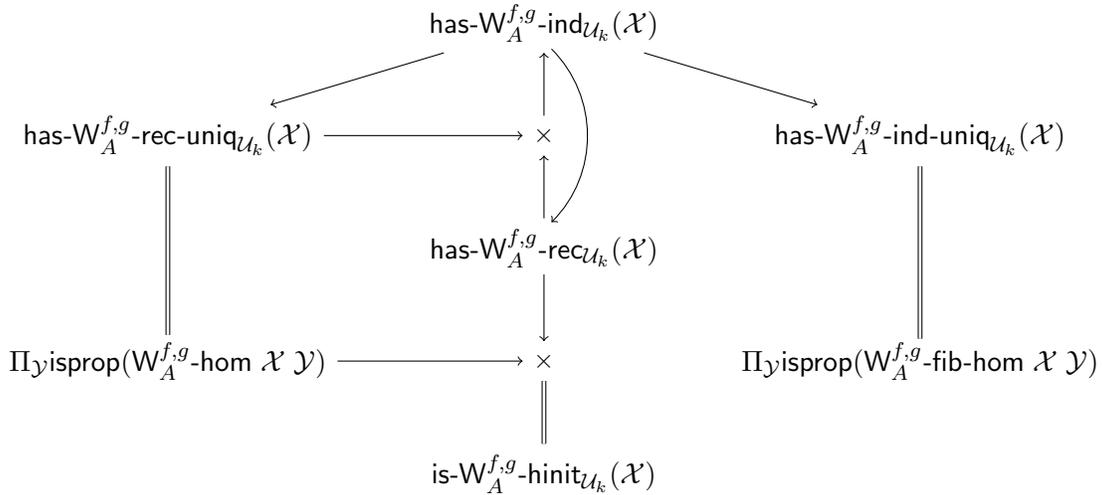
\begin{center}
\begin{tikzpicture}
\node (N10) at (7,-1.5) {\small $\hinitial{\U_k}(\mathcal{X})$};
\node (N11) at (7,4.5) {\small $\haswsuspind{\U_k}(\mathcal{X})$};
\node (N12) at (7,1.5) {\small $\haswsusprec{\U_k}(\mathcal{X})$};
\node (N13) at (2,3) {\small $\hasuniq{\U_k}(\mathcal{X})$};
\node (N14) at (12,3) {\small $\hasinduniq{\U_k}(\mathcal{X})$};
\node (N15) at (2,0) {\small $\prd{\mathcal{Y}} \isprop{\wsusphom\;\mathcal{X}\;\mathcal{Y}}$};
\node (N16) at (12,0) {\small $\prd{\mathcal{Y}} \isprop{\wsuspfibhom\;\mathcal{X}\;\mathcal{Y}}$};
\node (N15b) at (2,0.20) {};
\node (N16b) at (12,0.15) {};
\node (N17) at (7,3) {\footnotesize $\times$};
\node (N18) at (7,0) {\footnotesize $\times$};
\node (N11b) at (7.1,4.3) {};
\node (N12b) at (7.1,1.7) {};
\node (N12a) at (7,1.3) {};

\draw[->] (N11) -- node[above]{} (N14);
\draw[->] (N11) -- node[above]{} (N13);
\draw[->] (N17) -- node[above]{} (N11);
\draw[->] (N13) -- node[above]{} (N17);
\draw[->] (N12) -- node[above]{} (N17);
\draw[double distance = 1pt] (N14) -- node[above]{} (N16b);
\draw[double distance = 1pt] (N13) -- node[above]{} (N15b);
\draw[->] (N12a) -- node[above]{} (N18);
\draw[->] (N15) -- node[above]{} (N18);
\draw[double distance = 1pt] (N18) -- node[above]{} (N10);
\draw[->] (N11b.south) to[out=315,in=45] (N12b.north);
\end{tikzpicture}
\end{center}
Single arrow indicates implication; double line indicates equivalence. The symbol $\times$ indicates the product operator.

\section{Conclusion}\label{conc}
We have investigated higher inductive types with propositional computational behavior and shown that they can be equivalently characterized as homotopy-initial algebras. We have stated and proved this result for propositional truncations and for the so-called $\W$-suspensions, which subsume a number of other interesting cases - the unit circle $\mathbf{S}^1$, the interval type $\mathbf{I}$, all the higher spheres $\mathbf{S}^n$, and all suspensions. The characterization of these individual types as homotopy-initial algebras can be easily obtained as a corollary to our main theorem. Furthermore, we can readily apply the method presented here to obtain an analogous result for set truncations and set quotients. We conjecture that similar results can be established for other higher inductive types - such as homotopy (co)limits, tori, group quotients, or real numbers - following the same methodology. We are planning to formalize the results presented here in the Coq proof assistant.

Finally, we remark that the use of propositional computation rules instead of definitional ones alters the meaning of computation, which can now only be expressed up to a higher homotopy. The very same issue arises by postulating the univalence axiom itself, as well as any higher-dimensional constructors such as $\lp$. The precise computational interpretation of HoTT is currently a subject of intense research.

\section*{Acknowledgment}
The author would like to thank her advisors, Profs. Steve Awodey and Frank Pfenning, for their help.

\bibliographystyle{plain}
\bibliography{references}

\newpage
\appendix
\section{Proof of The Main Theorem}
\subsection{\emph{Preliminaries}}
We list here a few propositions which will be needed later. We omit the proofs as anyone reasonably familiar with HoTT should have no trouble verifying these statements.
\begin{proposition}\label{thm_cong_eq}
Given paths $u,v,w,z : b =_X c$ and $p : a =_X b$, and higher paths $\alpha : u = v$, $\beta : v = z$, $\gamma : u = w$, $\delta : w = z$, the commutativity of the diagram
\begin{center}
\begin{tikzpicture}
\node (N0) at (0,2) {$p \ct u$};
\node (N1) at (2,2) {$p \ct v$};
\node (N2) at (0,0) {$p \ct w$};
\node (N3) at (2,0) {$p \ct z$};
\draw[-] (N0) -- node[above]{\scriptsize \emph{via} $\alpha$} (N1);
\draw[-] (N0) -- node[left]{\scriptsize \emph{via} $\gamma$} (N2);
\draw[-] (N1) -- node[right]{\scriptsize \emph{via} $\beta$} (N3);
\draw[-] (N2) -- node[below]{\scriptsize \emph{via} $\delta$} (N3);
\end{tikzpicture}
\end{center}
is equivalent to the commutativity of
\begin{center}
\begin{tikzpicture}
\node (N0) at (0,2) {$u$};
\node (N1) at (2,2) {$v$};
\node (N2) at (0,0) {$w$};
\node (N3) at (2,0) {$z$};
\draw[-] (N0) -- node[above]{\scriptsize $\alpha$} (N1);
\draw[-] (N0) -- node[left]{\scriptsize $\gamma$} (N2);
\draw[-] (N1) -- node[right]{\scriptsize $\beta$} (N3);
\draw[-] (N2) -- node[below]{\scriptsize $\delta$} (N3);
\end{tikzpicture}
\end{center}
\end{proposition}

\begin{proposition}\label{thm_diag_rev}
Given paths $u : a =_X b$, $v : b =_X d$, $w : a =_X c$, $z : c =_X d$ and higher paths $\Phi, \Theta : u \ct v = w \ct z$, the commutativity of the diagram
\begin{center}
\begin{tikzpicture}
\node (N0) at (3,7) {$u$};
\node (N1) at (0,6) {$\refl{}{a} \ct u$};
\node (N2) at (6,6) {$u \ct \refl{}{b}$};
\node (N3) at (0,4) {$(w \ct \opp{w}) \ct u$};
\node (N4) at (6,4) {$u \ct (v \ct \opp{v})$};
\node (N5) at (0,2) {$w \ct (\opp{w} \ct u$)};
\node (N6) at (6,2) {$(u \ct v) \ct \opp{v}$};
\node (N7) at (0,0) {$w \ct (z \ct \opp{v})$};
\node (N8) at (6,0) {$(w \ct z) \ct \opp{v}$};
\draw[-] (N0) -- node[above]{} (N1);
\draw[-] (N0) -- node[left]{} (N2);
\draw[-] (N1) -- node[left]{} (N3);
\draw[-] (N2) -- node[left]{} (N4);
\draw[-] (N3) -- node[left]{} (N5);
\draw[-] (N4) -- node[left]{} (N6);
\draw[-] (N5) -- node[left]{\scriptsize \emph{via} $\invsq(\Phi)$} (N7);
\draw[-] (N6) -- node[right]{\scriptsize \emph{via} $\Theta$} (N8);
\draw[-] (N7) -- node[left]{} (N8);
\end{tikzpicture}
\end{center}
is equivalent to the path space $\Phi = \Theta$.
\end{proposition}


\subsection{\emph{Induction implies recursion}}\label{pf_ind_imp_rec}
Fix an algebra $(X,p,s) : \wsuspalg{\U_j}$ and assume that $\haswsuspind{\U_k}(X,p,s)$ holds. To show that $\haswsusprec{\U_k}(X,p,s)$ holds, fix any other algebra $(Y,q,r) : \wsuspalg{\U_k}$. In order to apply the induction principle, we need to turn this into a fibered algebra $(E,e,d)$. The first two components are easy: put $E \defeq \lam{\_:X} Y$ and $e \defeq q$. For the last component, we note that the transport between any two fibers of a constant type family is constant. We can thus define $d(b,a)$ to be the path
\begin{center}
\begin{tikzpicture}
\node (N0) at (0,0) {$\trans{\lam{\_:X} Y}{s(b,a)}(q(f\;b))$};
\node (N1) at (3.7,0) {$q(f\;b)$};
\node (N2) at (6.3,0) {$q(g\;b)$};
\draw[-] (N0) -- node[above]{} (N1);
\draw[-] (N1) -- node[above]{\footnotesize $r(b,a)$} (N2);
\end{tikzpicture}
\end{center}
The induction principle then gives us a map $h : X \to Y$ and a path family $\beta : \prd{c:C} (h(p(c)) = q(c))$ such that the following diagram commutes for any $b:B$, $a:B(a)$:

\begin{center}
\begin{tikzpicture}
\node (Nc) at (3.5,1) {=};
\node (N0) at (0,2) {$\trans{\lam{\_:X} Y}{s(b,a)}(h(p(f\;b)))$};
\node (N1) at (7,2) {$h(p(g\;b))$};
\node (N2) at (0,0) {$\trans{\lam{\_:X} Y}{s(b,a)}(q(f\;b))$};
\node (N3) at (4,0) {$q(f\;b)$};
\node (N4) at (7,0) {$q(g\;b)$};
\draw[-] (N0) -- node[above]{\footnotesize $\dap{h}{s(b,a)}$} (N1);
\draw[-] (N0) -- node[left]{\footnotesize $\ap{\trans{\lam{\_:X} Y}{s(b,a)}}{\beta(f\;b)}$} (N2);
\draw[-] (N2) -- node[above]{} (N3);
\draw[-] (N3) -- node[below]{\footnotesize $r(b,a)$} (N4);
\draw[-] (N1) -- node[right]{\footnotesize $\beta(g\;b)$} (N4);
\end{tikzpicture}
\end{center}
Using path induction we can express $\dap{h}{s(b,a)}$ equivalently as the path

\begin{center}
\begin{tikzpicture}
\node (N0) at (-0.1,0) {$\trans{\lam{\_:X} Y}{s(b,a)}(h(p(f\;b)))$};
\node (N1) at (4.2,0) {$h(p(f\;b))$};
\node (N2) at (8,0) {$h(p(g\;b))$};
\draw[-] (N0) -- node[above]{} (N1);
\draw[-] (N1) -- node[above]{\footnotesize $\ap{h}{s(b,a)}$} (N2);
\end{tikzpicture}
\end{center}
Thus the outer rectangle in the following diagram commutes:

\begin{center}
\begin{tikzpicture}
\node (NA) at (2.15,1) {\emph{A}};
\node (NB) at (6.1,1) {\emph{B}};
\node (N0) at (-0.1,2) {$\trans{\lam{\_:X} Y}{s(b,a)}(h(p(f\;b)))$};
\node (N1a) at (4.2,2) {$h(p(f\;b))$};
\node (N1) at (8,2) {$h(p(g\;b))$};
\node (N2) at (-0.1,0) {$\trans{\lam{\_:X} Y}{s(b,a)}(q(f\;b))$};
\node (N3) at (4.2,0) {$q(f\;b)$};
\node (N4) at (8,0) {$q(g\;b)$};
\draw[-] (N0) -- node[above]{} (N1a);
\draw[-] (N1a) -- node[above]{\footnotesize $\ap{h}{s(b,a)}$} (N1);
\draw[-] (N1a) -- node[left]{\footnotesize $\beta(f\;b)$} (N3);
\draw[-] (N0) -- node[left]{\footnotesize $\ap{\trans{\lam{\_:X} Y}{s(b,a)}}{\beta(f\;b)}$} (N2);
\draw[-] (N2) -- node[above]{} (N3);
\draw[-] (N3) -- node[below]{\footnotesize $r(b,a)$} (N4);
\draw[-] (N1) -- node[right]{\footnotesize $\beta(g\;b)$} (N4);
\end{tikzpicture}
\end{center}
Suitable path induction shows that rectangle \emph{A} commutes; hence rectangle \emph{B} commutes too and we are done.


\subsection{\emph{Induction implies uniqueness}}\label{pf_ind_imp_coh}
Fix an algebra $(X,p,s) : \wsuspalg{\U_j}$ and assume that $\haswsuspind{\U_k}(X,p,s)$ holds. We first show the dependent case, i.e., that $\hasinduniq{\U_k}(X,p,s)$ holds. Fix any fibered algebra $(E,e,d) : \wsuspfibalg{\U_k} \; (X,p,s)$ and homomorphisms
$(h,\gamma,\Theta), (i,\delta,\Phi) : \wsuspfibhom \; (X,p,s) \; (E,e,d)$. To construct a homotopy between $(h,\gamma,\Theta)$ and $(i,\delta,\Phi)$, we first need a homotopy $\alpha : h \sim i$. We can obtain $\alpha$ from the induction principle applied to a suitable fibered algebra of the form $(\lam{x:X} (h(x) = i(x)),e',d')$. The term $e'$ thus must be of the type $\prd{c:C}(h(p\;c) = i(p\;c))$. This is easy to get since we know how the maps $h$ and $i$ behave on constructors: we put $e'(c) \defeq \gamma(c) \ct \delta(c)^{-1}$. Finally, for $b:B$, $a:A(b)$, the term $d'(b,a)$ must be a path from $\trans{\lam{x:X} (h(x) = i(x))}{s(b,a)}(\gamma_{f(b)} \ct \delta_{f(b)}^{-1})$ to $\gamma_{g(b)} \ct \delta_{g(b)}^{-1}$. We note that for any $u : x =_X y$ and $v : h(x) = i(x)$, the transport $\trans{\lam{x:X} (h(x) = i(x))}{u}(v)$ can be expressed as $\opp{\dap{h}{u}} \ct (\ap{\trans{E}{u}}{v} \ct \dap{i}{u})$. We can thus define $d'(b,a)$ to be the path

\begin{center}
\begin{tikzpicture}
\node (N0) at (0,8) {$\trans{\lam{x:X} (h(x) = i(x))}{s(b,a)}(\gamma_{f(b)} \ct \delta_{f(b)}^{-1})$};
\node (N1) at (0,6.4) {$\opp{\dap{h}{s(b,a)}} \ct (\ap{\trans{E}{s(b,a)}}{\gamma_{f(b)} \ct \delta_{f(b)}^{-1}} \ct \dap{i}{s(b,a)})$};
\node (N2) at (0,4.8) {$\opp{\dap{h}{s(b,a)}} \ct (\dap{h}{s(b,a)} \ct (\gamma_{g(b)} \ct \delta_{g(b)}^{-1}))$};
\node (N3) at (0,3.2) {$(\opp{\dap{h}{s(b,a)}} \ct \dap{h}{s(b,a)}) \ct (\gamma_{g(b)} \ct \delta_{g(b)}^{-1})$};
\node (N4) at (0,1.6) {$\reflsym \ct (\gamma_{g(b)} \ct \delta_{g(b)}^{-1})$};
\node (N5) at (0,0) {$\gamma_{g(b)}) \ct \delta_{g(b)}^{-1}$};
\draw[-] (N0) -- node[above]{} (N1);
\draw[-] (N1) -- node[right]{\scriptsize via $\mathcal{H}$} (N2);
\draw[-] (N2) -- node[above]{} (N3);
\draw[-] (N3) -- node[above]{} (N4);
\draw[-] (N4) -- node[above]{} (N5);
\end{tikzpicture}
\end{center}
where $\mathcal{H}$ is the lower part of the diagram in Not.~\ref{nat_coh_dep}, i.e., the path
\begin{center}
\begin{tikzpicture}
\node (N1a) at (0,3.2) {$\ap{\trans{E}{s(b,a)}}{\gamma_{f(b)} \ct \delta^{-1}_{f(b)}} \ct \dap{i}{s(b,a)}$};
\node (N3)  at (0,1.6) {$(\ap{\trans{E}{s(b,a)}}{\gamma_{f(b)}} \ct \ap{\trans{E}{s(b,a)}}{\delta_{f(b)}}^{-1}) \ct \dap{i}{s(b,a)}\;\;\;\;\;\;\;\;\;\;\;\;\;\;\;\;\;\;\;\;$};
\node (N5)  at (0,0) {$\ap{\trans{E}{s(b,a)}}{\gamma_{f(b)}} \ct (\ap{\trans{E}{s(b,a)}}{\delta_{f(b)}}^{-1} \ct \dap{i}{s(b,a)})\;\;\;\;\;\;\;\;\;\;\;\;\;\;\;\;\;\;\;\;$};
\node (N7)  at (0,-1.6) {$\ap{\trans{E}{s(b,a)}}{\gamma_{f(b)}} \ct (d(b,a) \ct \delta^{-1}_{g(b)})$};

\node (N4)  at (8,1.6) {$\dap{h}{s(b,a)} \ct (\gamma_{g(b)} \ct \delta^{-1}_{g(b)})$};
\node (N6)  at (8,0) {$(\dap{h}{s(b,a)} \ct \gamma_{g(b)}) \ct \delta^{-1}_{g(b)}$};
\node (N8)  at (8,-1.6) {$(\ap{\trans{E}{s(b,a)}}{\gamma_{f(b)}} \ct d(b,a)) \ct \delta^{-1}_{g(b)}$};
\draw[-] (N1a) -- node[left]{\scriptsize} (N3);
\draw[-] (N3) -- node[above]{\footnotesize} (N5);
\draw[-] (N4) -- node[above]{\footnotesize} (N6);
\draw[-] (N5) -- node[left]{\scriptsize via $\invsq(\Phi(a,b))$} (N7);
\draw[-] (N6) -- node[right]{\scriptsize via $\Theta(a,b)$} (N8);
\draw[-] (N7) -- node[above]{\footnotesize} (N8);
\end{tikzpicture}
\end{center}
The induction principle then gives us $\alpha : h \sim i$ as desired; moreover, the first computation rule gives us a path family $\eta : \prd{c:C} (\alpha_{p(c)} = \gamma_c \ct \delta^{-1}_c)$. Thus $(\alpha, \lam{c:C} \invtri^{-1}(\eta_c)) : (h,\gamma,\Theta) \sim (i,\delta,\Phi)$. All that remains to show now is that the following diagram commutes for each $b:B$, $a:B(a)$:
\begin{center}
\begin{tikzpicture}
\node (N1)  at (0,6) {$\ap{\trans{E}{s(b,a)}}{\alpha_{p(f \; b)}} \ct \dap{i}{s(b,a)}$};
\node (N1a) at (0,4) {$\ap{\trans{E}{s(b,a)}}{\gamma_{f(b)} \ct \delta^{-1}_{f(b)}} \ct \dap{i}{s(b,a)}$};
\node (N2)  at (8,6) {$\dap{h}{s(b,a)} \ct \alpha_{p(g\;b)}$};
\node (N4)  at (8,4) {$\dap{h}{s(b,a)} \ct (\gamma_{g(b)} \ct \delta^{-1}_{g(b)})$};

\draw[-] (N1) -- node[above]{\scriptsize naturality of $\alpha$} (N2);
\draw[-] (N1) -- node[left]{\scriptsize via $\invtri(\invtri^{-1}(\eta_{f(b)}))$} (N1a);
\draw[-] (N2) -- node[right]{\scriptsize via $\invtri(\invtri^{-1}(\eta_{g(b)}))$} (N4);
\draw[-] (N1a) -- node[below]{\scriptsize $\mathcal{H}$} (N4);
\end{tikzpicture}
\end{center}     
The commutativity of the above diagram is equivalent to the commutativity of
\begin{center}
\begin{tikzpicture}
\node (N1)  at (0,6) {$\ap{\trans{E}{s(b,a)}}{\alpha_{p(f \; b)}} \ct \dap{i}{s(b,a)}$};
\node (N1a) at (0,4) {$\ap{\trans{E}{s(b,a)}}{\gamma_{f(b)} \ct \delta^{-1}_{f(b)}} \ct \dap{i}{s(b,a)}$};
\node (N2)  at (8,6) {$\dap{h}{s(b,a)} \ct \alpha_{p(g\;b)}$};
\node (N4)  at (8,4) {$\dap{h}{s(b,a)} \ct (\gamma_{g(b)} \ct \delta^{-1}_{g(b)})$};

\draw[-] (N1) -- node[above]{\scriptsize naturality of $\alpha$} (N2);
\draw[-] (N1) -- node[left]{\footnotesize {\scriptsize via} $\eta_{f(b)}$} (N1a);
\draw[-] (N2) -- node[right]{\footnotesize {\scriptsize via} $\eta_{g(b)}$} (N4);
\draw[-] (N1a) -- node[below]{\scriptsize $\mathcal{H}$} (N4);
\end{tikzpicture}
\end{center}     
To show this, we use the second computation rule, which tells us that the diagram below commutes for any $b:B$, $a:B(a)$:
\begin{center}
\begin{tikzpicture}
\node (N0a) at (4.5,5) {=};
\node (N0a) at (0,8) {$\trans{\lam{x:X} (h(x) = i(x))}{s(b,a)}(\alpha_{p(f\;b)})$};
\node (N5a) at (0,0) {$\alpha_{p(g\;b)}$};

\node (N0) at (9,8) {$\trans{\lam{x:X} (h(x) = i(x))}{s(b,a)}(\gamma_{f(b)} \ct \delta_{f(b)}^{-1})$};
\node (N1) at (9,6.4) {$\opp{\dap{h}{s(b,a)}} \ct (\ap{\trans{E}{s(b,a)}}{\gamma_{f(b)} \ct \delta_{f(b)}^{-1}} \ct \dap{i}{s(b,a)})$};
\node (N2) at (9,4.8) {$\opp{\dap{h}{s(b,a)}} \ct (\dap{h}{s(b,a)} \ct (\gamma_{g(b)} \ct \delta_{g(b)}^{-1}))$};
\node (N3) at (9,3.2) {$(\opp{\dap{h}{s(b,a)}} \ct \dap{h}{s(b,a)}) \ct (\gamma_{g(b)} \ct \delta_{g(b)}^{-1})$};
\node (N4) at (9,1.6) {$\reflsym \ct (\gamma_{g(b)} \ct \delta_{g(b)}^{-1})$};
\node (N5) at (9,0) {$\gamma_{g(b)} \ct \delta_{g(b)}^{-1}$};
\draw[-] (N0a) -- node[above]{{\scriptsize via} \footnotesize $\eta_{f(b)}$} (N0);
\draw[-] (N5a) -- node[below]{\footnotesize $\eta_{g(b)}$} (N5);
\draw[-] (N0a) -- node[left]{\footnotesize $\dap{\alpha}{s(b,a)}$} (N5a);

\draw[-] (N0) -- node[above]{} (N1);
\draw[-] (N1) -- node[right]{\scriptsize via $\mathcal{H}$} (N2);
\draw[-] (N2) -- node[above]{} (N3);
\draw[-] (N3) -- node[above]{} (N4);
\draw[-] (N4) -- node[above]{} (N5);
\end{tikzpicture}
\end{center}
We observe that for any $u : x =_X y$, we can express $\dap{\alpha}{u}$ as the path
\begin{center}
\begin{tikzpicture}
\node (N0) at (0,8) {$\trans{\lam{x:X} (h(x) = i(x))}{u}(\alpha_x)$};
\node (N1) at (0,6.4) {$\opp{\dap{h}{u}} \ct (\ap{\trans{E}{u}}{\alpha_x} \ct \dap{i}{u})$};
\node (N2) at (0,4.8) {$\opp{\dap{h}{u}} \ct (\dap{h}{u} \ct \alpha_y)$};
\node (N3) at (0,3.2) {$(\opp{\dap{h}{u}} \ct \dap{h}{u}) \ct \alpha_y$};
\node (N4) at (0,1.6) {$\reflsym \ct \alpha_y$};
\node (N5) at (0,0) {$\alpha_y$};
\draw[-] (N0) -- node[above]{} (N1);
\draw[-] (N1) -- node[right]{\scriptsize via naturality of $\alpha$} (N2);
\draw[-] (N2) -- node[above]{} (N3);
\draw[-] (N3) -- node[above]{} (N4);
\draw[-] (N4) -- node[above]{} (N5);
\end{tikzpicture}
\end{center}
Thus, the second computation rule may be expressed as saying that the outer parallellogram in the diagram below commutes:
\begin{center}
\begin{tikzpicture}
\node (NA) at (4,8) {\emph{A}};
\node (NA) at (4,4.5) {\emph{B}};
\node (NA) at (4,1.5) {\emph{C}};
\node (NA) at (4,-1.5) {\emph{D}};
\node (NA) at (4,-4.158) {\emph{E}};

\node (N0a) at (0,10) {$\trans{\lam{x:X} (h(x) = i(x))}{s(b,a)}(\alpha_{p(f\;b)})$};
\node (N1a) at (0,7) {\small $\opp{\dap{h}{s(b,a)}} \ct (\ap{\trans{E}{s(b,a)}}{\alpha_{p(f\;b)}} \ct \dap{i}{s(b,a)})$};
\node (N2a) at (0,4) {\small $\opp{\dap{h}{s(b,a)}} \ct (\dap{h}{s(b,a)} \ct \alpha_{p(g\;b)})$ \;\;\;\;\;\;\;\;\;\;\;\;\;\;};
\node (N3a) at (0,1) {\small $(\opp{\dap{h}{s(b,a)}} \ct \dap{h}{s(b,a)}) \ct \alpha_{p(g\;b)}$ \;\;\;\;\;\;\;\;\;\;\;\;\;\;};
\node (N4a) at (0,-2) {\small $\reflsym \ct \alpha_{p(g\;b)}$};
\node (N5a) at (0,-5) {$\alpha_{p(g\;b)}$};

\node (N0) at (7.7,8.5) {$\trans{\lam{x:X} (h(x) = i(x))}{s(b,a)}(\gamma_{f(b)} \ct \delta_{f(b)}^{-1})$};
\node (N1) at (7.7,5.5) {\small $\opp{\dap{h}{s(b,a)}} \ct (\ap{\trans{E}{s(b,a)}}{\gamma_{f(b)} \ct \delta_{f(b)}^{-1}} \ct \dap{i}{s(b,a)})$};
\node (N2) at (7.7,2.5) {\small $\opp{\dap{h}{s(b,a)}} \ct (\dap{h}{s(b,a)} \ct (\gamma_{g(b)} \ct \delta_{g(b)}^{-1}))$};
\node (N3) at (7.7,-0.5) {\small $(\opp{\dap{h}{s(b,a)}} \ct \dap{h}{s(b,a)}) \ct (\gamma_{g(b)} \ct \delta_{g(b)}^{-1})$};
\node (N4) at (7.7,-3.5) {\small $\reflsym \ct (\gamma_{g(b)} \ct \delta_{g(b)}^{-1})$};
\node (N5) at (7.7,-6.5) {$\gamma_{g(b)} \ct \delta_{g(b)}^{-1}$};
\draw[-] (N0a) -- node[above]{\;\;\;\;\;\;\;\;\;\;\;\;{\scriptsize via} \scriptsize $\eta_{f(b)}$} (N0);
\draw[-] (N1a) -- node[above]{\;\;\;\;\;\;\;\;\;\;\;\;{\scriptsize via} \scriptsize $\eta_{f(b)}$} (N1);
\draw[-] (N2a) -- node[above]{\;\;\;\;\;\;\;\;\;\;\;\;{\scriptsize via} \scriptsize $\eta_{g(b)}$} (N2);
\draw[-] (N3a) -- node[above]{\;\;\;\;\;\;\;\;\;\;\;\;{\scriptsize via} \scriptsize $\eta_{g(b)}$} (N3);
\draw[-] (N4a) -- node[above]{\;\;\;\;\;\;\;\;\;\;\;\;{\scriptsize via} \scriptsize $\eta_{g(b)}$} (N4);
\draw[-] (N5a) -- node[above]{\;\;\;\;\;\;\;\;\;\;\;\;{\scriptsize via} \scriptsize $\eta_{g(b)}$} (N5);
\draw[-] (N0a) -- node[above]{} (N1a);
\draw[-] (N1a) -- node[left]{\scriptsize via naturality of $\alpha$} (N2a);
\draw[-] (N2a) -- node[above]{} (N3a);
\draw[-] (N3a) -- node[above]{} (N4a);
\draw[-] (N4a) -- node[above]{} (N5a);

\draw[-] (N0) -- node[above]{} (N1);
\draw[-] (N1) -- node[right]{\scriptsize via $\mathcal{H}$} (N2);
\draw[-] (N2) -- node[above]{} (N3);
\draw[-] (N3) -- node[above]{} (N4);
\draw[-] (N4) -- node[above]{} (N5);
\end{tikzpicture}
\end{center}
We can easily show that the parallellograms \emph{A, C, D, E} commute. Thus \emph{B} commutes as well. By Prop.~\ref{thm_cong_eq} we conclude that the following diagram commutes
\begin{center}
\begin{tikzpicture}
\node (N1)  at (0,6) {$\ap{\trans{E}{s(b,a)}}{\alpha_{p(f \; b)}} \ct \dap{i}{s(b,a)}$};
\node (N1a) at (8,6) {$\ap{\trans{E}{s(b,a)}}{\gamma_{f(b)} \ct \delta^{-1}_{f(b)}} \ct \dap{i}{s(b,a)}$};
\node (N2)  at (0,4) {$\dap{h}{s(b,a)} \ct \alpha_{p(g\;b)}$};
\node (N4)  at (8,4) {$\dap{h}{s(b,a)} \ct (\gamma_{g(b)} \ct \delta^{-1}_{g(b)})$};

\draw[-] (N1) -- node[left]{\scriptsize naturality of $\alpha$} (N2);
\draw[-] (N1) -- node[above]{\footnotesize {\scriptsize via} $\eta_{f(b)}$} (N1a);
\draw[-] (N2) -- node[below]{\footnotesize {\scriptsize via} $\eta_{g(b)}$} (N4);
\draw[-] (N1a) -- node[right]{\scriptsize $\mathcal{H}$} (N4);
\end{tikzpicture}
\end{center}
which is precisely what we wanted to show. \bigskip

The non-dependent case, i.e., showing $\hasuniq{\U_k}(X,p,s)$, proceeds by an entirely analogous argument, further simplified by the fact that we no longer need to transport along the fibers of the codomain type $E$. 

\emph{Remark:} With some effort, we could obtain the non-dependent case from the result we have just proved. However, due to the presence of superfluous transports, it is much simpler to establish the non-dependent result directly, following the same methodology.


\subsection{\emph{Characterizing the path space of homomorphisms}}\label{pf_two_cell_path}
We first cover the dependent case: for any fibered homomorphisms $\mu,\nu$, the path space $\mu = \nu$ is equivalent to $\twofibcell \; \mu \; \nu$. To show this, fix an algebra $(X,p,s) : \wsuspalg{\U_j}$ and a fibered algebra $(E,e,d) : \wsuspfibalg{\U_k} \; (X,p,s)$. A homomorphism from $(X,p,s)$ to $(E,e,d)$ is thus a triple $(h,\gamma,\Theta)$ as given in Def.~\ref{def_fibhom}. In this section, it will be more useful for us to consider the representation $((h,\gamma),\Theta)$ instead, i.e., associated to the left rather than to the right. The pair $(h,\gamma) : \poinfun{C}{E}{p}{e}$ then represents a pointed function. For convenience, we also name the type of $\Theta$: we define a coherence condition on $\poinfun{C}{E}{p}{e}$ by
\[ \poinfuncoh{s}{d}(h,\gamma) \defeq \prd{b:B}\prd{a:A(b)} (\dap{h}{s(b,a)} \ct \gamma(g\;b) = \ap{\trans{E}{s(b,a)}}{\beta(f\;b)} \ct d(b,a)) \]
Homomorphisms from $(X,p,s)$ to $(E,e,d)$ are precisely those pointed maps satisfying the coherence condition:
\[ \wsuspfibhom^{\text{L}} \; (X,p,s) \; (E,e,d) \defeq \sm{(\theta: \poinfun{C}{E}{p}{e})} \poinfuncoh{s}{d}(\theta) \]
We likewise have the ``left-associated" versions of $\mu \sim \nu$, $\wsuspfibcoh$, and $\twofibcell$:
\begin{align*}
(\theta,\Theta) \sim_{\text{L}}(\phi,\Phi) & \defeq \theta \sim \phi \\
\wsuspfibcoh^\text{L} \; b \; a \; ((h,\gamma),\Theta) \; ((i,\delta),\Phi) \; \mathfrak{p} & \defeq \wsuspfibcoh \; b \; a \; (h,\gamma,\Theta) \; (i,\delta,\Phi) \; \mathfrak{p} \\
\twofibcell^\text{L} \; \mu \; \nu & \defeq \sm{(\mathfrak{p} : \mu \sim_{\text{L}} \nu)} \prd{b:B} \prd{a:A(b)} \wsuspfibcoh^{\text{L}} \; b \; a \; \mu \; \nu \; \mathfrak{p}
\end{align*}
It now suffices to show that for any homomorphisms $\mu,\nu : \wsuspfibhom^{\text{L}} \; (X,p,s) \; (E,e,d)$, we have $(\mu = \nu) \simeq \twofibcell^\text{L} \; \mu \; \nu$. Fix two such homomorphisms $(\theta,\Theta)$ and $(\phi,\Phi)$. We have
\begin{align}\label{eq}
((\theta,\Theta) = (\phi,\Phi)) \simeq \sm{(\mathfrak{p} : \theta = \phi)} (\Theta = \transc{\poinfuncoh{s}{d}}{\mathfrak{p}}(\Phi))
\end{align}
It is easy to see that the path space $\theta = \phi$ is equivalent to the space of homotopies $\theta \sim \phi$: for any $(h,\gamma)$ and $(i,\delta)$ we have the chain of equivalences
\begin{alignat*}{4}
& (h,\gamma) = (i,\delta) & \; \; \; \simeq \\
& \sm{(\alpha : h = i)} (\gamma = \transc{\lam{j} \prd{c:C} (j(p(c)) = e(c))}{\alpha} \delta) & \; \; \; \simeq \\
& \sm{(\alpha : h = i)} (\gamma = \lam{c:C} \happly{h}{i}(\alpha)(p(c)) \ct \delta(c)) & \; \; \; \simeq \\
& \sm{(\alpha : h = i)}\prd{c:C} (\gamma(c) = \happly{h}{i}(\alpha)(p(c)) \ct \delta(c)) & \; \; \; \simeq \\
& \sm{(\alpha : h \sim i)} \prd{c:C} (\gamma(c) = \alpha(p(c)) \ct \delta(c)) & \; \; \; \equiv \\
& (h,\gamma) \sim (i,\delta) &
\end{alignat*}
Let $\mathcal{P}_{\theta,\phi} : (\theta = \phi) \to (\theta \sim \phi)$ denote the composition of these equivalences. We now show that for any $\theta,\phi,\mathfrak{p} : \theta = \phi$, $\Theta$, $\Phi$, we have
\begin{align*} (\Theta = \transc{\poinfuncoh{s}{d}}{\mathfrak{p}}(\Phi)) \simeq \prd{b:B}\prd{a:A} \wsuspfibcoh^{\text{L}} \; b \; a \; (\theta,\Theta) \; (\phi,\Phi) \; \mathcal{P}_{\theta,\phi}(\mathfrak{p}) \tag{$\star$} \end{align*}
We proceed by path induction on $\mathfrak{p}$. We thus need to show that for any $\theta$, $\Theta$, $\Phi$, we have
\[ (\Theta = \Phi) \simeq \prd{b:B}\prd{a:A} \wsuspfibcoh^{\text{L}} \; b \; a \; (\theta,\Theta) \; (\theta,\Phi) \; \mathcal{P}_{\theta,\theta}(\refl{}{\theta}) \]
It suffices to show that for any $h,\gamma,b,a$, we have
\[ (\Theta(b,a) = \Phi(b,a)) \simeq \wsuspfibcoh^{\text{L}} \; b \; a \; ((h,\gamma),\Theta) \; ((h,\gamma),\Phi) \; \mathcal{P}_{(h,\gamma),(h,\gamma)}(\refl{}{h,\gamma}) \]
It is easy to show that $\mathcal{P}_{(h,\gamma),(h,\gamma)}(\refl{}{h,\gamma}) = (\alpha, \eta)$, where $\alpha$ is the identity homotopy on $h$ and $\eta$ assigns to each $c:C$ the path
\begin{center}
\begin{tikzpicture}
\node (N0) at (0,0) {$\gamma(c)$};
\node (N1) at (2,0) {$\reflsym \ct \gamma(c)$};
\draw[-] (N0) -- node[above]{} (N1);
\end{tikzpicture}
\end{center}
We thus need to show that the path space $\Theta(b,a) = \Phi(b,a)$ is equivalent to the commutativity of the following diagram:
\begin{center}
\begin{tikzpicture}
\node (N1)  at (0,4.8) {$\reflsym \ct \dap{h}{s(b,a)}$};
\node (N1a) at (0,3.2) {$\ap{\trans{E}{s(b,a)}}{\gamma_{f(b)} \ct \gamma^{-1}_{f(b)}} \ct \dap{h}{s(b,a)}$};
\node (N3)  at (0,1.6) {$(\ap{\trans{E}{s(b,a)}}{\gamma_{f(b)}} \ct \ap{\trans{E}{s(b,a)}}{\gamma_{f(b)}}^{-1}) \ct \dap{h}{s(b,a)}\;\;\;\;\;\;\;\;\;\;\;\;\;\;\;\;\;\;\;\;$};
\node (N5)  at (0,0) {$\ap{\trans{E}{s(b,a)}}{\gamma_{f(b)}} \ct (\ap{\trans{E}{s(b,a)}}{\gamma_{f(b)}}^{-1} \ct \dap{h}{s(b,a)})\;\;\;\;\;\;\;\;\;\;\;\;\;\;\;\;\;\;\;\;$};
\node (N7)  at (0,-1.6) {$\ap{\trans{E}{s(b,a)}}{\gamma_{f(b)}} \ct (d(b,a) \ct \gamma^{-1}_{g(b)})$};

\node (N2)  at (8,4.8) {$\dap{h}{s(b,a)} \ct \reflsym$};
\node (N4)  at (8,1.6) {$\dap{h}{s(b,a)} \ct (\gamma_{g(b)} \ct \gamma^{-1}_{g(b)})$};
\node (N6)  at (8,0) {$(\dap{h}{s(b,a)} \ct \gamma_{g(b)}) \ct \gamma^{-1}_{g(b)}$};
\node (N8)  at (8,-1.6) {$(\ap{\trans{E}{s(b,a)}}{\gamma_{f(b)}} \ct d(b,a)) \ct \gamma^{-1}_{g(b)}$};
\draw[-] (N1) -- node[above]{\scriptsize \emph{naturality of} $\alpha$} (N2);
\draw[-] (N1) -- node[left]{\scriptsize \emph{via} $\invtri(\eta_{f(b)})$} (N1a);
\draw[-] (N1a) -- node[left]{\scriptsize} (N3);
\draw[-] (N2) -- node[right]{\scriptsize \emph{via} $\invtri(\eta_{g(b)})$} (N4);
\draw[-] (N3) -- node[above]{\footnotesize} (N5);
\draw[-] (N4) -- node[above]{\footnotesize} (N6);
\draw[-] (N5) -- node[left]{\scriptsize \emph{via} $\invsq(\Phi(a,b))$} (N7);
\draw[-] (N6) -- node[right]{\scriptsize \emph{via} $\Theta(a,b)$} (N8);
\draw[-] (N7) -- node[above]{\footnotesize} (N8);
\end{tikzpicture}
\end{center}
Expressing $\invtri(\eta_{f(b)})$, $\invtri(\eta_{g(b)})$, and the naturality of $\alpha$ directly yields the diagram:
\begin{center}
\begin{tikzpicture}
\node (N0)  at (4,6.4) {$\dap{h}{s(b,a)}$};
\node (N1)  at (0,4.8) {$\reflsym \ct \dap{h}{s(b,a)}$};
\node (N1a) at (0,3.2) {$\ap{\trans{E}{s(b,a)}}{\gamma_{f(b)} \ct \gamma^{-1}_{f(b)}} \ct \dap{h}{s(b,a)}$};
\node (N3)  at (0,1.6) {$(\ap{\trans{E}{s(b,a)}}{\gamma_{f(b)}} \ct \ap{\trans{E}{s(b,a)}}{\gamma_{f(b)}}^{-1}) \ct \dap{h}{s(b,a)}\;\;\;\;\;\;\;\;\;\;\;\;\;\;\;\;\;\;\;\;$};
\node (N5)  at (0,0) {$\ap{\trans{E}{s(b,a)}}{\gamma_{f(b)}} \ct (\ap{\trans{E}{s(b,a)}}{\gamma_{f(b)}}^{-1} \ct \dap{h}{s(b,a)})\;\;\;\;\;\;\;\;\;\;\;\;\;\;\;\;\;\;\;\;$};
\node (N7)  at (0,-1.6) {$\ap{\trans{E}{s(b,a)}}{\gamma_{f(b)}} \ct (d(b,a) \ct \gamma^{-1}_{g(b)})$};

\node (N2)  at (8,4.8) {$\dap{h}{s(b,a)} \ct \reflsym$};
\node (N4)  at (8,1.6) {$\dap{h}{s(b,a)} \ct (\gamma_{g(b)} \ct \gamma^{-1}_{g(b)})$};
\node (N6)  at (8,0) {$(\dap{h}{s(b,a)} \ct \gamma_{g(b)}) \ct \gamma^{-1}_{g(b)}$};
\node (N8)  at (8,-1.6) {$(\ap{\trans{E}{s(b,a)}}{\gamma_{f(b)}} \ct d(b,a)) \ct \gamma^{-1}_{g(b)}$};
\draw[red,-] (N0) -- node[left]{\scriptsize} (N1);
\draw[red,-] (N0) -- node[left]{\scriptsize} (N2);
\draw[red,-] (N1) -- node[left]{\scriptsize} (N1a);
\draw[-] (N1a) -- node[left]{\scriptsize} (N3);
\draw[red,-] (N2) -- node[right]{\scriptsize} (N4);
\draw[-] (N3) -- node[above]{\footnotesize} (N5);
\draw[-] (N4) -- node[above]{\footnotesize} (N6);
\draw[-] (N5) -- node[left]{\scriptsize \emph{via} $\invsq(\Phi(a,b))$} (N7);
\draw[-] (N6) -- node[right]{\scriptsize \emph{via} $\Theta(a,b)$} (N8);
\draw[-] (N7) -- node[above]{\footnotesize} (N8);
\end{tikzpicture}
\end{center}
After some simplification we get the diagram
\begin{center}
\begin{tikzpicture}
\node (N0)  at (4,4.8) {$\dap{h}{s(b,a)}$};
\node (N1)  at (0,3.2) {$\reflsym \ct \dap{h}{s(b,a)}$};
\node (N3)  at (0,1.6) {$(\ap{\trans{E}{s(b,a)}}{\gamma_{f(b)}} \ct \ap{\trans{E}{s(b,a)}}{\gamma_{f(b)}}^{-1}) \ct \dap{h}{s(b,a)}\;\;\;\;\;\;\;\;\;\;\;\;\;\;\;\;\;\;\;\;$};
\node (N5)  at (0,0) {$\ap{\trans{E}{s(b,a)}}{\gamma_{f(b)}} \ct (\ap{\trans{E}{s(b,a)}}{\gamma_{f(b)}}^{-1} \ct \dap{h}{s(b,a)})\;\;\;\;\;\;\;\;\;\;\;\;\;\;\;\;\;\;\;\;$};
\node (N7)  at (0,-1.6) {$\ap{\trans{E}{s(b,a)}}{\gamma_{f(b)}} \ct (d(b,a) \ct \gamma^{-1}_{g(b)})$};

\node (N2)  at (8,3.2) {$\dap{h}{s(b,a)} \ct \reflsym$};
\node (N4)  at (8,1.6) {$\dap{h}{s(b,a)} \ct (\gamma_{g(b)} \ct \gamma^{-1}_{g(b)})$};
\node (N6)  at (8,0) {$(\dap{h}{s(b,a)} \ct \gamma_{g(b)}) \ct \gamma^{-1}_{g(b)}$};
\node (N8)  at (8,-1.6) {$(\ap{\trans{E}{s(b,a)}}{\gamma_{f(b)}} \ct d(b,a)) \ct \gamma^{-1}_{g(b)}$};
\draw[-] (N0) -- node[left]{\scriptsize} (N1);
\draw[-] (N0) -- node[left]{\scriptsize} (N2);
\draw[red,-] (N1) -- node[left]{\scriptsize} (N3);
\draw[-] (N2) -- node[right]{\scriptsize} (N4);
\draw[-] (N3) -- node[above]{\footnotesize} (N5);
\draw[-] (N4) -- node[above]{\footnotesize} (N6);
\draw[-] (N5) -- node[left]{\scriptsize \emph{via} $\invsq(\Phi(a,b))$} (N7);
\draw[-] (N6) -- node[right]{\scriptsize \emph{via} $\Theta(a,b)$} (N8);
\draw[-] (N7) -- node[above]{\footnotesize} (N8);
\end{tikzpicture}
\end{center}
By Prop.~\ref{thm_diag_rev}, the commutativity of this diagram is equivalent to the path space $\Phi(b,a) = \Theta(b,a)$, which is of course equivalent to $\Theta(b,a) = \Phi(b,a)$. This proves the claim $(\star)$. \medskip

We thus have
\begin{align*}
\big(\sm{(\mathfrak{p} : \theta = \phi)} (\Theta = \transc{\poinfuncoh{s}{d}}{\mathfrak{p}}(\Phi))\big) & \simeq \sm{(\mathfrak{p} : \theta = \phi)} \prd{b:B}\prd{a:A} \wsuspfibcoh^{\text{L}} \; b \; a \; (\theta,\Theta) \; (\phi,\Phi) \; \mathcal{P}_{\theta,\phi}(\mathfrak{p}) \\
& \simeq \sm{(\mathfrak{p} : \theta \sim \phi)} \prd{b:B}\prd{a:A} \wsuspfibcoh^{\text{L}} \; b \; a \; (\theta,\Theta) \; (\phi,\Phi) \; \mathfrak{p} \\
& \equiv \twofibcell^\text{L} \; (\theta,\Theta) \; (\phi,\Phi)
\end{align*}
which together with \ref{eq} finishes the proof. \bigskip

The non-dependent case, i.e., showing that $(\mu = \nu) \simeq \twocell \; \mu \; \nu$ for any $\mu,\nu$, follows by an entirely analogous argument.


\subsection{\emph{Recursion plus uniqueness imply induction}}\label{pf_rec_coh_imp_ind}
Fix $(X,p,s) : \wsuspalg{\U_j}$. Assume that $\haswsusprec{\U_k}(X,p,s)$ and $\hasuniq{\U_k}(X,p,s)$ hold. To show that $\haswsuspind{\U_k}(X,p,s)$ holds, fix any $(E,e,d) : \wsuspfibalg{\U_k}\; (X,p,s)$. In order to apply the recursion principle, we need to turn this into a non-fibered algebra $(Y,q,r)$. The first two components are easy: we put $Y \defeq \sm{x:X}E(x)$ and $q \defeq \lam{c:C} (p(c), e(c))$. We note that since $X : \U_j$, $E : X \to \U_k$, and $j \leq k$, we indeed have $\sm{x:X}E(x) : \U_k$ as needed. Finally,
we note that by Prop.~\ref{thm_pair_equiv} there is a function $\dpath{c}{d} : \big(\sm{(p :\fst{c} = \fst{d})} (\trans{B}{p}(\snd{c}) = \snd{d})\big) \to (c = d)$ for any $c,d$, which forms a quasi-equivalence with $\dpair{c}{d}$. We can thus define $r(b,a)$ to be the path
\begin{center}
\begin{tikzpicture}
\node (N0) at (0,0) {$(p(f\;b), e(f\;b))$};
\node (N1) at (7.5,0) {$(p(g\;b), e(g\;b))$};
\draw[-] (N0) -- node[above]{\footnotesize $\dpathsym(s(b,a),d(b,a))$} (N1);
\end{tikzpicture}
\end{center}
where the subscripts to $\dpathsym$ are omitted. The recursion principle thus gives us a function $u : X \to \sm{x:X} E(x)$. We now want to construct a homotopy $\alpha : \mathsf{fst} \circ u \sim \idfun{X}$. We can obtain $\alpha$ from the uniqueness principle applied to the algebra $(X,p,s)$ itself and homomorphisms of the form $(\fstsym \circ u, \gamma, \Theta)$, $(\idfun{X}, \delta, \Phi)$. Finding suitable $\delta$ and $\Phi$ is easy: we let $\delta(c) \defeq \refl{X}{p(c)}$ and $\Phi(b,a)$ to be the path
\begin{center}
\begin{tikzpicture}
\node (N0)  at (0,4.5) {$\ap{\idfun{X}}{s(b,a)} \ct \reflsym$};
\node (N1)  at (0,3) {$\ap{\idfun{X}}{s(b,a)}$};
\node (N2)  at (0,1.5) {$s(b,a)$};
\node (N3)  at (0,0) {$\reflsym \ct s(b,a)$};

\draw[-] (N0) -- node[left]{\scriptsize} (N1);
\draw[-] (N1) -- node[left]{\scriptsize} (N2);
\draw[-] (N2) -- node[right]{\scriptsize} (N3);
\end{tikzpicture}
\end{center}
The path family $\gamma$ should assign to each $c:C$ a path from $\fst{u(p(c))}$ to $p(c)$. The first computation rule for $u$ gives us a path family $\beta : \prd{c:C} (u(p(c)) = (p(c),e(c)))$. We can thus define $\gamma(c) \defeq \fst{\dpairsym(\beta(c))}$. Before we define $\Theta$, we make a few general observations that will be useful later on: \medskip

Let $l,m,n : \sm{x:X} E(x)$. Then:
\begin{enumerate}
\item For any $\epsilon : m = n$ we have $\fst{\dpairsym(\epsilon)} = \ap{\fstsym}{\epsilon}$.
\item\label{obs_tri} For any $y,z : \sm{(p : \fst{m} = \fst{n})}(\trans{E}{p}(\snd{m}) = \snd{n})$ and $\epsilon : y = z$, the following diagram commutes:
\begin{center}
\begin{tikzpicture}
\node (Na)  at (2.5,1.25) {=};
\node (N0)  at (0,2.5) {$\trans{E}{\fst{y}}(\snd{m})$};
\node (N1)  at (5,2.5) {$\trans{E}{\fst{z}}(\snd{m})$};
\node (N2)  at (2.5,0) {$\snd{n}$};

\draw[-] (N0) -- node[above]{\scriptsize via \footnotesize $\epsilon$} (N1);
\draw[-] (N0) -- node[left]{\footnotesize $\snd{y}\;$} (N2);
\draw[-] (N1) -- node[right]{\footnotesize $\;\snd{z}$} (N2);
\end{tikzpicture}
\end{center}
\item\label{obs_exp} For any $y : l = m$ and $z : m = n$, the following diagram commutes:
\begin{center}
\begin{tikzpicture}
\node (Na)  at (4,4) {=};
\node (N0)  at (0,8) {$\trans{E}{\ap{\fstsym}{y \ct z}}(\snd{l})$};
\node (N1)  at (0,0) {$\trans{E}{\fst{\dpairsym(y \ct z)}}(\snd{l})$};
\node (N2)  at (8,0) {$\snd{n}$};
\node (N3)  at (8,8) {$\trans{E}{(\ap{\fstsym}{y} \ct \ap{\fstsym}{z})}(\snd{l})$};
\node (N4)  at (8,6.4) {$\trans{E}{\ap{\fstsym}{z}}(\trans{E}{\ap{\fstsym}{y}}(\snd{l}))$};
\node (N5)  at (8,4.8) {$\trans{E}{\ap{\fstsym}{z}}(\trans{E}{\fst{\dpairsym(y)}}(\snd{l}))$};
\node (N6)  at (8,3.2) {$\trans{E}{\ap{\fstsym}{z}}(\snd{m})$};
\node (N7)  at (8,1.6) {$\trans{E}{\fst{\dpairsym(z)}}(\snd{m})$};

\draw[-] (N0) -- node[left]{} (N1);
\draw[-] (N0) -- node[left]{} (N3);
\draw[-] (N1) -- node[below]{\footnotesize $\snd{\dpairsym(y \ct z)}$} (N2);
\draw[-] (N3) -- node[left]{} (N4);
\draw[-] (N4) -- node[left]{} (N5);
\draw[-] (N5) -- node[right]{\scriptsize via \footnotesize $\snd{\dpairsym(y)}$} (N6);
\draw[-] (N6) -- node[right]{} (N7);
\draw[-] (N7) -- node[right]{\footnotesize $\snd{\dpairsym(z)}$} (N2);
\end{tikzpicture}
\end{center}
\end{enumerate}
\medskip

By the second computation rule for $u$, the following diagram commutes for each $b,a$:
\begin{center}
\begin{tikzpicture}
\node (Naa)  at (3.75,1) {$(\star \star)$};
\node (N0)  at (0,2) {$u(p(f\;b))$};
\node (N1)  at (7.5,2) {$u(p(g\;b))$};
\node (N2)  at (0,0) {$(p(f\;b),e(f\;b))$};
\node (N3)  at (7.5,0) {$(p(g\;b),e(g\;b))$};

\draw[-] (N0) -- node[above]{\footnotesize $\ap{u}{s(b,a)}$} (N1);
\draw[-] (N0) -- node[left]{\footnotesize $\beta(f(b))$} (N2);
\draw[-] (N1) -- node[right]{\footnotesize $\beta(g(b))$} (N3);
\draw[-] (N2) -- node[below]{\footnotesize $\dpathsym(s(b,a),d(b,a))$} (N3);
\end{tikzpicture}
\end{center}
We now define $\Theta(b,a)$ as the following path:
\begin{center}
\begin{tikzpicture}
\node (N0)  at (0,9.6) {$\ap{\fstsym \circ u}{s(b,a)} \ct \fst{\dpairsym(\beta_{g(b)})}$};
\node (N1)  at (0,8)  {$\ap{\fstsym \circ u}{s(b,a)} \ct \ap{\fstsym}{\beta_{g(b)}}$ };
\node (N2)  at (0,6.4)  {$\ap{\fstsym}{\ap{u}{s(b,a)}} \ct \ap{\fstsym}{\beta_{g(b)}}$};
\node (N3)  at (0,4.8)  {$\ap{\fstsym}{\ap{u}{s(b,a)} \ct \beta_{g(b)}}$};
\node (N4)  at (0,3.2)  {$\ap{\fstsym}{\beta_{f(b)} \ct \dpathsym(s(b,a), d(b,a))}$};
\node (N5)  at (0,1.6)  {$\ap{\fstsym}{\beta_{f(b)}} \ct \ap{\fstsym}{\dpathsym(s(b,a), d(b,a))}$};
\node (N6)  at (0,0)    {$\ap{\fstsym}{\beta_{f(b)}} \ct \fst{\dpairsym(\dpathsym(s(b,a), d(b,a)))}$};
\node (N7)  at (0,-1.6) {$\ap{\fstsym}{\beta_{f(b)}} \ct s(b,a)$};
\node (N8)  at (0,-3.2) {$\fst{\dpairsym(\beta_{f(b)})} \ct s(b,a)$};

\draw[-] (N0) -- node[left]{\scriptsize} (N1);
\draw[-] (N1) -- node[left]{\scriptsize} (N2);
\draw[-] (N2) -- node[right]{\scriptsize} (N3);
\draw[-] (N3) -- node[right]{\scriptsize via $(\star\star)$} (N4);
\draw[-] (N4) -- node[right]{\scriptsize} (N5);
\draw[-] (N5) -- node[right]{\scriptsize} (N6);
\draw[-] (N6) -- node[right]{\scriptsize} (N7);
\draw[-] (N7) -- node[right]{\scriptsize} (N8);
\end{tikzpicture}
\end{center}
The uniqueness rule thus gives us the desired homotopy $\alpha : \fstsym \circ u \sim \idfun{X}$ together with a path family $\eta : \prd{c:C} \big(\fst{\dpairsym(\beta_c)} = \alpha_{p(c)} \ct \reflsym\big)$. \medskip

We can now define the inductor $h(x) \defeq \trans{E}{(\alpha_x)}(\snd{u_x})$. To establish the first computation rule, we need a path family $\beta_D$ assigning to each $c:C$ a path from $\trans{E}{(\alpha_{p(c)})}(\snd{u_{p(c)}})$ to $e(c)$. This is relatively easy: we let $\beta_D(c)$ be the path
\begin{center}
\begin{tikzpicture}
\node (N0)  at (0,4.2) {$\trans{E}{(\alpha_{p(c)})}(\snd{u_{p(c)}})$};
\node (N1)  at (0,2.8) {$\trans{E}{(\alpha_{p(c)} \ct \reflsym)}(\snd{u_{p(c)}})$};
\node (N2)  at (0,1.4) {$\trans{E}{\fst{\dpairsym(\beta_c)}}(\snd{u_{p(c)}})$};
\node (N3)  at (0,0) {$e(c)$};

\draw[-] (N0) -- node[right]{} (N1);
\draw[-] (N1) -- node[right]{\scriptsize via \footnotesize $\eta_c$} (N2);
\draw[-] (N2) -- node[right]{\footnotesize $\snd{\dpairsym(\beta_c)}$} (N3);
\end{tikzpicture}
\end{center}
To establish the second computation rule, we need to show that the following diagram commutes for each $b,a$:
\begin{center}
\begin{tikzpicture}
\node (N0)  at (0,2) {$\trans{E}{s(b,a)}(\trans{E}{(\alpha_{p(f\;b)})}(\snd{u_{p(f\;b)}}))$};
\node (N1)  at (9.5,2) {$\trans{E}{(\alpha_{p(g\;b)})}(\snd{u_{p(g\;b)}})$};
\node (N2)  at (0,0) {$\trans{E}{s(b,a)}(e(f\;b))$};
\node (N3)  at (9.5,0) {$e(g\;b)$};

\draw[-] (N0) -- node[above]{\footnotesize $\dap{\trans{E}{\alpha(-)}(\snd{u(-)})}{s(b,a)}$} (N1);
\draw[-] (N1) -- node[right]{\footnotesize $\beta_D(g\;b)$} (N3);
\draw[-] (N0) -- node[left]{\footnotesize $\ap{\trans{E}{s(b,a)}}{\beta_D(f\;b)}$} (N2);
\draw[-] (N2) -- node[below]{\footnotesize $d(b,a)\;\;$} (N3);
\end{tikzpicture}
\end{center}
This requires a significant amount of work and will be done in 3 parts. Part I and II simplify each of the respective paths around the above diagram; part III then shows these paths are equal.

\paragraph{Part I} We first simplify the path $\ap{\trans{E}{s(b,a)}}{\beta_D(f\;b)} \ct d(b,a)$. Expanding, we get
\begin{center}
\begin{tikzpicture}
\node (N0)  at (0,4.8) {$\trans{E}{s(b,a)}(\trans{E}{(\alpha_{p(f\;b)})}(\snd{u_{p(f\;b)}}))$};
\node (N1)  at (0,3.2) {$\trans{E}{s(b,a)}(\trans{E}{(\alpha_{p(f\;b)} \ct \reflsym)}(\snd{u_{p(f\;b)}}))$};
\node (N2)  at (0,1.6) {$\trans{E}{s(b,a)}(\trans{E}{\fst{\dpairsym(\beta_{f(b)})}}(\snd{u_{p(f\;b)}}))$};
\node (N3)  at (0,0) {$\trans{E}{s(b,a)}(e(f\;b))$};
\node (N4)  at (0,-1.6) {$e(g\;b)$};

\draw[-] (N0) -- node[right]{} (N1);
\draw[-] (N1) -- node[right]{\scriptsize via \footnotesize $\eta_{f(b)}$} (N2);
\draw[-] (N2) -- node[right]{\scriptsize via \footnotesize $\snd{\dpairsym(\beta_{f(b)})}$} (N3);
\draw[-] (N3) -- node[right]{\footnotesize $d(b,a)$} (N4);
\end{tikzpicture}
\end{center}
which can be further expanded to
\begin{center}
\begin{tikzpicture}
\node (N0)  at (0,4.8) {$\trans{E}{s(b,a)}(\trans{E}{(\alpha_{p(f\;b)})}(\snd{u_{p(f\;b)}}))$};
\node (N1)  at (0,3.2) {$\trans{E}{s(b,a)}(\trans{E}{(\alpha_{p(f\;b)} \ct \reflsym)}(\snd{u_{p(f\;b)}}))$};
\node (N2)  at (0,1.6) {$\trans{E}{s(b,a)}(\trans{E}{\fst{\dpairsym(\beta_{f(b)})}}(\snd{u_{p(f\;b)}}))$};
\node[red] (N2a)  at (0,0) {$\trans{E}{s(b,a)}(\trans{E}{\ap{\fstsym}{\beta_{f(b)}}}(\snd{u_{p(f\;b)}}))$};
\node (N2b)  at (0,-1.6) {$\trans{E}{s(b,a)}(\trans{E}{\fst{\dpairsym(\beta_{f(b)})}}(\snd{u_{p(f\;b)}}))$};
\node (N3)  at (0,-3.2) {$\trans{E}{s(b,a)}(e(f\;b))$};
\node (N4)  at (0,-4.8) {$e(g\;b)$};

\draw[-] (N0) -- node[right]{} (N1);
\draw[-] (N1) -- node[right]{\scriptsize via \footnotesize $\eta_{f(b)}$} (N2);
\draw[red,-] (N2) -- node[right]{} (N2a);
\draw[red,-] (N2a) -- node[right]{} (N2b);
\draw[-] (N2b) -- node[right]{\scriptsize via \footnotesize $\snd{\dpairsym(\beta_{f(b)})}$} (N3);
\draw[-] (N3) -- node[right]{\footnotesize $d(b,a)$} (N4);
\end{tikzpicture}
\end{center}
which is equal to
\begin{center}
\begin{tikzpicture}
\node (N0)  at (0,4.8) {$\trans{E}{s(b,a)}(\trans{E}{(\alpha_{p(f\;b)})}(\snd{u_{p(f\;b)}}))$};
\node[red] (N0a)  at (0,3.2) {$\trans{E}{(\alpha_{p(f\;b)} \ct s(b,a))}(\snd{u_{p(f\;b)}})$};
\node[red] (N1)  at (0,1.6) {$\trans{E}{((\alpha_{p(f\;b)} \ct \reflsym) \ct s(b,a))}(\snd{u_{p(f\;b)}})$};
\node[red] (N2)  at (0,0) {$\trans{E}{(\fst{\dpairsym(\beta_{f(b)})} \ct s(b,a))}(\snd{u_{p(f\;b)}})$};
\node[red] (N2a)  at (0,-1.6) {$\trans{E}{(\ap{\fstsym}{\beta_{f(b)}} \ct s(b,a))}(\snd{u_{p(f\;b)}})$};
\node (N2b)  at (0,-3.2) {$\trans{E}{s(b,a)}(\trans{E}{\ap{\fstsym}{\beta_{f(b)}}}(\snd{u_{p(f\;b)}}))$};
\node (N2c)  at (0,-4.8) {$\trans{E}{s(b,a)}(\trans{E}{\fst{\dpairsym(\beta_{f(b)})}}(\snd{u_{p(f\;b)}}))$};
\node (N3)  at (0,-6.4) {$\trans{E}{s(b,a)}(e(f\;b))$};
\node (N4)  at (0,-8) {$e(g\;b)$};

\draw[red,-] (N0) -- node[right]{} (N0a);
\draw[red,-] (N0a) -- node[right]{} (N1);
\draw[red,-] (N1) -- node[right]{\scriptsize via \footnotesize $\eta_{f(b)}$} (N2);
\draw[red,-] (N2) -- node[right]{} (N2a);
\draw[red,-] (N2a) -- node[right]{} (N2b);
\draw[-] (N2b) -- node[right]{} (N2c);
\draw[-] (N2c) -- node[right]{\scriptsize via \footnotesize $\snd{\dpairsym(\beta_{f(b)})}$} (N3);
\draw[-] (N3) -- node[right]{\footnotesize $d(b,a)$} (N4);
\end{tikzpicture}
\end{center}
By Obs.~\ref{obs_tri} the following diagram commutes:
\begin{center}
\begin{tikzpicture}
\node (Na)  at (4,1.75) {=};
\node (N0)  at (0,3.5) {$\trans{E}{\fst{\dpairsym(\dpathsym(s(b,a),d(b,a)))}}e(f\;b))$};
\node (N1)  at (8,3.5) {$\trans{E}{s(b,a)}(e(f\;b))$};
\node (N2)  at (4,0) {$e(g\;b)$};

\draw[-] (N0) -- node[above]{} (N1);
\draw[-] (N0) -- node[left]{\footnotesize $\snd{\dpairsym(\dpathsym(s(b,a),d(b,a)))}\;\;$} (N2);
\draw[-] (N1) -- node[right]{\footnotesize $\;\;d(b,a)$} (N2);
\end{tikzpicture}
\end{center}
The above path is thus equal to
\begin{center}
\begin{tikzpicture}
\node (N0)  at (0,4.8) {$\trans{E}{s(b,a)}(\trans{E}{(\alpha_{p(f\;b)})}(\snd{u_{p(f\;b)}}))$};
\node (N0a)  at (0,3.2) {$\trans{E}{(\alpha_{p(f\;b)} \ct s(b,a))}(\snd{u_{p(f\;b)}})$};
\node (N1)  at (0,1.6) {$\trans{E}{((\alpha_{p(f\;b)} \ct \reflsym) \ct s(b,a))}(\snd{u_{p(f\;b)}})$};
\node (N2)  at (0,0) {$\trans{E}{(\fst{\dpairsym(\beta_{f(b)})} \ct s(b,a))}(\snd{u_{p(f\;b)}})$};
\node (N2a)  at (0,-1.6) {$\trans{E}{(\ap{\fstsym}{\beta_{f(b)}} \ct s(b,a))}(\snd{u_{p(f\;b)}})$};
\node[red] (N2ab)  at (0,-3.2) {$\trans{E}{\big(\ap{\fstsym}{\beta_{f(b)}} \ct \fst{\dpairsym(\dpathsym(s(b,a),d(b,a))})\big)}(\snd{u_{p(f\;b)}})$};
\node[red] (N2ac)  at (0,-4.8) {$\trans{E}{\big(\ap{\fstsym}{\beta_{f(b)}} \ct \ap{\fstsym}{\dpathsym(s(b,a),d(b,a))}\big)}(\snd{u_{p(f\;b)}})$};
\node[red](N2b)  at (0,-6.4) {$\trans{E}{\ap{\fstsym}{\dpathsym(s(b,a),d(b,a))}}(\trans{E}{\ap{\fstsym}{\beta_{f(b)}}}(\snd{u_{p(f\;b)}}))$};
\node[red](N2c)  at (0,-8) {$\trans{E}{\ap{\fstsym}{\dpathsym(s(b,a),d(b,a))}}(\trans{E}{\fst{\dpairsym(\beta_{f(b)})}}(\snd{u_{p(f\;b)}}))$};
\node[red](N3)  at (0,-9.6) {$\trans{E}{\ap{\fstsym}{\dpathsym(s(b,a),d(b,a))}}(e(f\;b))$};
\node[red](N3a)  at (0,-11.2) {$\trans{E}{\fst{\dpairsym(\dpathsym(s(b,a),d(b,a)))}}(e(f\;b))$};
\node (N4)  at (0,-12.8) {$e(g\;b)$};

\draw[-] (N0) -- node[right]{} (N0a);
\draw[-] (N0a) -- node[right]{} (N1);
\draw[-] (N1) -- node[right]{\scriptsize via \footnotesize $\eta_{f(b)}$} (N2);
\draw[-] (N2) -- node[right]{} (N2a);
\draw[red,-] (N2a) -- node[right]{} (N2ab);
\draw[red,-] (N2ab) -- node[right]{} (N2ac);
\draw[red,-] (N2ac) -- node[right]{} (N2b);
\draw[red,-] (N2b) -- node[right]{} (N2c);
\draw[red,-] (N2c) -- node[right]{\scriptsize via \footnotesize $\snd{\dpairsym(\beta_{f(b)})}$} (N3);
\draw[red,-] (N3) -- node[right]{} (N3a);
\draw[red,-] (N3a) -- node[right]{\footnotesize $\snd{\dpairsym(\dpathsym(s(b,a),d(b,a)))}$} (N4);
\end{tikzpicture}
\end{center}
By Obs.~\ref{obs_exp} this is equal to
\begin{center}
\begin{tikzpicture}
\node (N0)  at (0,4.8) {$\trans{E}{s(b,a)}(\trans{E}{(\alpha_{p(f\;b)})}(\snd{u_{p(f\;b)}}))$};
\node (N0a)  at (0,3.2) {$\trans{E}{(\alpha_{p(f\;b)} \ct s(b,a))}(\snd{u_{p(f\;b)}})$};
\node (N1)  at (0,1.6) {$\trans{E}{((\alpha_{p(f\;b)} \ct \reflsym) \ct s(b,a))}(\snd{u_{p(f\;b)}})$};
\node (N2)  at (0,0) {$\trans{E}{(\fst{\dpairsym(\beta_{f(b)})} \ct s(b,a))}(\snd{u_{p(f\;b)}})$};
\node (N2a)  at (0,-1.6) {$\trans{E}{(\ap{\fstsym}{\beta_{f(b)}} \ct s(b,a))}(\snd{u_{p(f\;b)}})$};
\node (N2ab)  at (0,-3.2) {$\trans{E}{\big(\ap{\fstsym}{\beta_{f(b)}} \ct \fst{\dpairsym(\dpathsym(s(b,a),d(b,a)))}\big)}(\snd{u_{p(f\;b)}})$};
\node (N2ac)  at (0,-4.8) {$\trans{E}{\big(\ap{\fstsym}{\beta_{f(b)}} \ct \ap{\fstsym}{\dpathsym(s(b,a),d(b,a))}\big)}(\snd{u_{p(f\;b)}})$};
\node[red](N2b)  at (0,-6.4) {$\trans{E}{\ap{\fstsym}{\beta_{f(b)} \ct \dpathsym(s(b,a),d(b,a))}}(\snd{u_{p(f\;b)}})$};
\node[red](N2c)  at (0,-8) {$\trans{E}{\fst{\dpairsym(\beta_{f(b)} \ct \dpathsym(s(b,a),d(b,a)))}}(\snd{u_{p(f\;b)}})$};
\node (N4)  at (0,-9.6) {$e(g\;b)$};

\draw[-] (N0) -- node[right]{} (N0a);
\draw[-] (N0a) -- node[right]{} (N1);
\draw[-] (N1) -- node[right]{\scriptsize via \footnotesize $\eta_{f(b)}$} (N2);
\draw[-] (N2) -- node[right]{} (N2a);
\draw[-] (N2a) -- node[right]{} (N2ab);
\draw[-] (N2ab) -- node[right]{} (N2ac);
\draw[red,-] (N2ac) -- node[right]{} (N2b);
\draw[red,-] (N2b) -- node[right]{} (N2c);
\draw[red,-] (N2c) -- node[right]{\footnotesize $\snd{\dpairsym(\beta_{f(b)} \ct \dpathsym(s(b,a),d(b,a)))}$} (N4);
\end{tikzpicture}
\end{center}

\paragraph{Part II}
We now simplify the path $\dap{\trans{E}{\alpha(-)}(\snd{u(-)})}{s(b,a)} \ct \beta_D(g\;b)$. It is not hard to see that for any $\epsilon : x =_X y$, the path $\dap{\trans{E}{\alpha(-)}(\snd{u(-)})}{\epsilon}$ can be expressed explicitly as the path
\begin{center}
\begin{tikzpicture}
\node (N0)  at (0,4.8) {$\trans{E}{\epsilon}(\trans{E}{(\alpha_{x})}(\snd{u_{x}}))$};
\node (N0a)  at (0,3.2) {$\trans{E}{(\alpha_{x} \ct \epsilon)}(\snd{u_{x}})$};
\node (N1)  at (0,1.6) {$\trans{E}{(\alpha_{x} \ct \ap{\idfun{X}}{\epsilon})}(\snd{u_{x}})$};
\node (N2)  at (0,0) {$\trans{E}{(\ap{\fstsym \circ u}{\epsilon} \ct \alpha_{y})}(\snd{u_{x}})$};
\node (N2a)  at (0,-1.6) {$\trans{E}{(\ap{\fstsym}{\ap{u}{\epsilon}} \ct \alpha_{y})}(\snd{u_{x}})$};
\node (N2ab)  at (0,-3.2) {$\trans{E}{(\alpha_{y})}(\trans{E}{\ap{\fstsym}{\ap{u}{\epsilon}}}(\snd{u_{x}}))$};
\node (N2ac)  at (0,-4.8) {$\trans{E}{(\alpha_{y})}(\trans{E}{\fst{\dpairsym(\ap{u}{\epsilon})}}(\snd{u_{x}}))$};
\node (N2b)  at (0,-6.4) {$\trans{E}{(\alpha_{y})}(\snd{u_{y}})$};

\draw[-] (N0) -- node[right]{} (N0a);
\draw[-] (N0a) -- node[right]{} (N1);
\draw[-] (N1) -- node[right]{\scriptsize via naturality of \footnotesize $\alpha$} (N2);
\draw[-] (N2) -- node[right]{} (N2a);
\draw[-] (N2a) -- node[right]{} (N2ab);
\draw[-] (N2ab) -- node[right]{} (N2ac);
\draw[-] (N2ac) -- node[right]{\scriptsize via \footnotesize $\snd{\dpairsym(\ap{u}{\epsilon})}$} (N2b);
\end{tikzpicture}
\end{center}
The path $\dap{\trans{E}{\alpha(-)}(\snd{u(-)})}{s(b,a)} \ct \beta_D(g\;b)$ is thus equal to
\begin{center}
\begin{tikzpicture}
\node (N0)  at (0,4.8) {$\trans{E}{s(b,a)}(\trans{E}{(\alpha_{p(f\;b)})}(\snd{u_{p(f\;b)}}))$};
\node (N0a)  at (0,3.2) {$\trans{E}{(\alpha_{p(f\;b)} \ct s(b,a))}(\snd{u_{p(f\;b)}})$};
\node (N1)  at (0,1.6) {$\trans{E}{(\alpha_{p(f\;b)} \ct \ap{\idfun{X}}{s(b,a)})}(\snd{u_{p(f\;b)}})$};
\node (N2)  at (0,0) {$\trans{E}{(\ap{\fstsym \circ u}{s(b,a)} \ct \alpha_{p(g\;b)})}(\snd{u_{p(f\;b)}})$};
\node (N2a)  at (0,-1.6) {$\trans{E}{(\ap{\fstsym}{\ap{u}{s(b,a)}} \ct \alpha_{p(g\;b)})}(\snd{u_{p(f\;b)}})$};
\node (N2ab)  at (0,-3.2) {$\trans{E}{(\alpha_{p(g\;b)})}(\trans{E}{\ap{\fstsym}{\ap{u}{s(b,a)}}}(\snd{u_{p(f\;b)}}))$};
\node (N2ac)  at (0,-4.8) {$\trans{E}{(\alpha_{p(g\;b)})}(\trans{E}{\fst{\dpairsym(\ap{u}{s(b,a)})}}(\snd{u_{p(f\;b)}}))$};
\node (N2b)  at (0,-6.4) {$\trans{E}{(\alpha_{p(g\;b)})}(\snd{u_{p(g\;b)}})$};
\node (N3)  at (0,-8) {$\trans{E}{(\alpha_{p(g\;b)} \ct \reflsym)}(\snd{u_{p(g\;b)}})$};
\node (N4)  at (0,-9.6) {$\trans{E}{\fst{\dpairsym(\beta_{g(b)})}}(\snd{u_{p(g\;b)}})$};
\node (N5)  at (0,-11.2) {$e(g\;b)$};

\draw[-] (N0) -- node[right]{} (N0a);
\draw[-] (N0a) -- node[right]{} (N1);
\draw[-] (N1) -- node[right]{\scriptsize via naturality of \footnotesize $\alpha$} (N2);
\draw[-] (N2) -- node[right]{} (N2a);
\draw[-] (N2a) -- node[right]{} (N2ab);
\draw[-] (N2ab) -- node[right]{} (N2ac);
\draw[-] (N2ac) -- node[right]{\scriptsize via \footnotesize $\snd{\dpairsym(\ap{u}{s(b,a)})}$} (N2b);
\draw[-] (N2b) -- node[right]{} (N3);
\draw[-] (N3) -- node[right]{\scriptsize via \footnotesize $\eta_{g(b)}$} (N4);
\draw[-] (N4) -- node[right]{\footnotesize $\snd{\dpairsym(\beta_{g(b)})}$} (N5);
\end{tikzpicture}
\end{center}
Expanding further, we get
\begin{center}
\begin{tikzpicture}
\node (N0)  at (0,4.8) {$\trans{E}{s(b,a)}(\trans{E}{(\alpha_{p(f\;b)})}(\snd{u_{p(f\;b)}}))$};
\node (N0a)  at (0,3.2) {$\trans{E}{(\alpha_{p(f\;b)} \ct s(b,a))}(\snd{u_{p(f\;b)}})$};
\node (N1)  at (0,1.6) {$\trans{E}{(\alpha_{p(f\;b)} \ct \ap{\idfun{X}}{s(b,a)})}(\snd{u_{p(f\;b)}})$};
\node (N2)  at (0,0) {$\trans{E}{(\ap{\fstsym \circ u}{s(b,a)} \ct \alpha_{p(g\;b)})}(\snd{u_{p(f\;b)}})$};
\node (N2a)  at (0,-1.6) {$\trans{E}{(\ap{\fstsym}{\ap{u}{s(b,a)}} \ct \alpha_{p(g\;b)})}(\snd{u_{p(f\;b)}})$};
\node (N2ab)  at (0,-3.2) {$\trans{E}{(\alpha_{p(g\;b)})}(\trans{E}{\ap{\fstsym}{\ap{u}{s(b,a)}}}(\snd{u_{p(f\;b)}}))$};
\node (N2ac)  at (0,-4.8) {$\trans{E}{(\alpha_{p(g\;b)})}(\trans{E}{\fst{\dpairsym(\ap{u}{s(b,a)})}}(\snd{u_{p(f\;b)}}))$};
\node (N2b)  at (0,-6.4) {$\trans{E}{(\alpha_{p(g\;b)})}(\snd{u_{p(g\;b)}})$};
\node (N3)  at (0,-8) {$\trans{E}{(\alpha_{p(g\;b)} \ct \reflsym)}(\snd{u_{p(g\;b)}})$};
\node (N4)  at (0,-9.6) {$\trans{E}{\fst{\dpairsym(\beta_{g(b)})}}(\snd{u_{p(g\;b)}})$};
\node[red] (N4a)  at (0,-11.2) {$\trans{E}{\ap{\fstsym}{\beta_{g(b)}}}(\snd{u_{p(g\;b)}})$};
\node (N4b)  at (0,-12.8) {$\trans{E}{\fst{\dpairsym(\beta_{g(b)})}}(\snd{u_{p(g\;b)}})$};
\node (N5)  at (0,-14.4) {$e(g\;b)$};

\draw[-] (N0) -- node[right]{} (N0a);
\draw[-] (N0a) -- node[right]{} (N1);
\draw[-] (N1) -- node[right]{\scriptsize via naturality of \footnotesize $\alpha$} (N2);
\draw[-] (N2) -- node[right]{} (N2a);
\draw[-] (N2a) -- node[right]{} (N2ab);
\draw[-] (N2ab) -- node[right]{} (N2ac);
\draw[-] (N2ac) -- node[right]{\scriptsize via \footnotesize $\snd{\dpairsym(\ap{u}{s(b,a)})}$} (N2b);
\draw[-] (N2b) -- node[right]{} (N3);
\draw[-] (N3) -- node[right]{\scriptsize via \footnotesize $\eta_{g(b)}$} (N4);
\draw[red,-] (N4) -- node[right]{} (N4a);
\draw[red,-] (N4a) -- node[right]{} (N4b);
\draw[-] (N4b) -- node[right]{\footnotesize $\snd{\dpairsym(\beta_{g(b)})}$} (N5);
\end{tikzpicture}
\end{center}
which is equal to
\begin{center}
\begin{tikzpicture}
\node (N0)  at (0,4.8) {$\trans{E}{s(b,a)}(\trans{E}{(\alpha_{p(f\;b)})}(\snd{u_{p(f\;b)}}))$};
\node (N0a)  at (0,3.2) {$\trans{E}{(\alpha_{p(f\;b)} \ct s(b,a))}(\snd{u_{p(f\;b)}})$};
\node (N1)  at (0,1.6) {$\trans{E}{(\alpha_{p(f\;b)} \ct \ap{\idfun{X}}{s(b,a)})}(\snd{u_{p(f\;b)}})$};
\node (N2)  at (0,0) {$\trans{E}{(\ap{\fstsym \circ u}{s(b,a)} \ct \alpha_{p(g\;b)})}(\snd{u_{p(f\;b)}})$};
\node[red] (N2a)  at (0,-1.6) {$\trans{E}{(\ap{\fstsym \circ u}{s(b,a)} \ct (\alpha_{p(g\;b)} \ct \reflsym))}(\snd{u_{p(f\;b)}})$};
\node[red] (N2ab)  at (0,-3.2) {$\trans{E}{\big(\ap{\fstsym \circ u}{s(b,a)} \ct \fst{\dpairsym(\beta_{g(b)})}\big)}(\snd{u_{p(f\;b)}})$};
\node[red] (N2ac)  at (0,-4.8) {$\trans{E}{\big(\ap{\fstsym \circ u}{s(b,a)} \ct \ap{\fstsym}{\beta_{g(b)}}\big)}(\snd{u_{p(f\;b)}})$};
\node[red] (N2b)  at (0,-6.4) {$\trans{E}{(\ap{\fstsym}{\ap{u}{s(b,a)}} \ct \ap{\fstsym}{\beta_{g(b)}})}(\snd{u_{p(f\;b)}})$};
\node[red] (N3)  at (0,-8) {$\trans{E}{\ap{\fstsym}{\beta_{g(b)}}}(\trans{E}{\ap{\fstsym}{\ap{u}{s(b,a)}}}(\snd{u_{p(f\;b)}}))$};
\node[red] (N4)  at (0,-9.6) {$\trans{E}{\ap{\fstsym}{\beta_{g(b)}}}(\trans{E}{\fst{\dpairsym(\ap{u}{s(b,a)})}}(\snd{u_{p(f\;b)}}))$};
\node (N4a) at (0,-11.2) {$\trans{E}{\ap{\fstsym}{\beta_{g(b)}}}(\snd{u_{p(g\;b)}})$};
\node (N4b)  at (0,-12.8) {$\trans{E}{\fst{\dpairsym(\beta_{g(b)})}}(\snd{u_{p(g\;b)}})$};
\node (N5)  at (0,-14.4) {$e(g\;b)$};

\draw[-] (N0) -- node[right]{} (N0a);
\draw[-] (N0a) -- node[right]{} (N1);
\draw[-] (N1) -- node[right]{\scriptsize via naturality of \footnotesize $\alpha$} (N2);
\draw[red,-] (N2) -- node[right]{} (N2a);
\draw[red,-] (N2a) -- node[right]{\scriptsize via \footnotesize $\eta_{g(b)}$} (N2ab);
\draw[red,-] (N2ab) -- node[right]{} (N2ac);
\draw[red,-] (N2ac) -- node[right]{} (N2b);
\draw[red,-] (N2b) -- node[right]{} (N3);
\draw[red,-] (N3) -- node[right]{} (N4);
\draw[red,-] (N4) -- node[right]{\scriptsize via \footnotesize $\snd{\dpairsym(\ap{u}{s(b,a)})}$} (N4a);
\draw[-] (N4a) -- node[right]{} (N4b);
\draw[-] (N4b) -- node[right]{\footnotesize $\snd{\dpairsym(\beta_{g(b)})}$} (N5);
\end{tikzpicture}
\end{center}
By Obs.~\ref{obs_exp} this is equal to
\begin{center}
\begin{tikzpicture}
\node (N0)  at (0,4.8) {$\trans{E}{s(b,a)}(\trans{E}{(\alpha_{p(f\;b)})}(\snd{u_{p(f\;b)}}))$};
\node (N0a)  at (0,3.2) {$\trans{E}{(\alpha_{p(f\;b)} \ct s(b,a))}(\snd{u_{p(f\;b)}})$};
\node (N1)  at (0,1.6) {$\trans{E}{(\alpha_{p(f\;b)} \ct \ap{\idfun{X}}{s(b,a)})}(\snd{u_{p(f\;b)}})$};
\node (N2)  at (0,0) {$\trans{E}{(\ap{\fstsym \circ u}{s(b,a)} \ct \alpha_{p(g\;b)})}(\snd{u_{p(f\;b)}})$};
\node (N2a)  at (0,-1.6) {$\trans{E}{(\ap{\fstsym \circ u}{s(b,a)} \ct (\alpha_{p(g\;b)} \ct \reflsym))}(\snd{u_{p(f\;b)}})$};
\node (N2ab)  at (0,-3.2) {$\trans{E}{\big(\ap{\fstsym \circ u}{s(b,a)} \ct \fst{\dpairsym(\beta_{g(b)})}\big)}(\snd{u_{p(f\;b)}})$};
\node (N2ac)  at (0,-4.8) {$\trans{E}{\big(\ap{\fstsym \circ u}{s(b,a)} \ct \ap{\fstsym}{\beta_{g(b)}}\big)}(\snd{u_{p(f\;b)}})$};
\node (N2b)  at (0,-6.4) {$\trans{E}{(\ap{\fstsym}{\ap{u}{s(b,a)}} \ct \ap{\fstsym}{\beta_{g(b)}})}(\snd{u_{p(f\;b)}})$};
\node[red] (N3)  at (0,-8) {$\trans{E}{\ap{\fstsym}{\ap{u}{s(b,a)} \ct \beta_{g(b)}}}(\snd{u_{p(f\;b)}})$};
\node[red] (N4)  at (0,-9.6) {$\trans{E}{\fst{\dpairsym(\ap{u}{s(b,a)} \ct \beta_{g(b)})}}(\snd{u_{p(f\;b)}})$};
\node (N5)  at (0,-11.2) {$e(g\;b)$};

\draw[-] (N0) -- node[right]{} (N0a);
\draw[-] (N0a) -- node[right]{} (N1);
\draw[-] (N1) -- node[right]{\scriptsize via naturality of \footnotesize $\alpha$} (N2);
\draw[-] (N2) -- node[right]{} (N2a);
\draw[-] (N2a) -- node[right]{\scriptsize via \footnotesize $\eta_{g(b)}$} (N2ab);
\draw[-] (N2ab) -- node[right]{} (N2ac);
\draw[-] (N2ac) -- node[right]{} (N2b);
\draw[red,-] (N2b) -- node[right]{} (N3);
\draw[red,-] (N3) -- node[right]{} (N4);
\draw[red,-] (N4) -- node[right]{\footnotesize $\snd{\dpairsym(\ap{u}{s(b,a)} \ct \beta_{g(b)})}$} (N5);
\end{tikzpicture}
\end{center}

\paragraph{Part III}
By Obs.~\ref{obs_exp} the following diagram commutes:
\begin{center}
\begin{tikzpicture}
\node (Na)  at (2,2.25) {=};
\node (N0)  at (0,4.5) {$\trans{E}{\fst{\dpairsym\big(\beta_{f(b)} \ct \dpathsym(s(b,a),d(b,a))\big)}}(\snd{u_{p(f\;b)}})$};
\node (N1)  at (0,0) {$\trans{E}{\fst{\dpairsym\big(\ap{u}{s(b,a)} \ct \beta_{g(b)}\big)}}(\snd{u_{p(f\;b)}})$\;\;\;\;\;\;\;\;\;\;\;\;};
\node (N2)  at (6,2.25) {$e(g\;b)$};

\draw[-] (N0) -- node[left]{\scriptsize via $(\star\star)$} (N1);
\draw[-] (N0) -- node[right]{\footnotesize \;\;\;\;\;$\snd{\dpairsym\big(\beta_{f(b)} \ct \dpathsym(s(b,a),d(b,a))\big)}$} (N2);
\draw[-] (N1) -- node[right]{\footnotesize $\;\;\;\;\;\snd{\dpairsym\big(\ap{u}{s(b,a)} \ct \beta_{g(b)}\big)}$} (N2);
\end{tikzpicture}
\end{center}
It thus suffices to show that the following diagram commutes:
\begin{center}
\begin{tikzpicture}
\node (Na)  at (4.5,4.8) {$(\star\star\star)$};
\node (N0)  at (0,9.6) {$\alpha_{p(f\;b)} \ct s(b,a)$};
\node (N1)  at (0,8)   {$(\alpha_{p(f\;b)} \ct \reflsym) \ct s(b,a)$};
\node (N2)  at (0,6.4) {$\fst{\dpairsym(\beta_{f(b)})} \ct s(b,a)$};
\node (N3)  at (0,4.8) {$\ap{\fstsym}{\beta_{f(b)}} \ct s(b,a)$};
\node (N4)  at (0,3.2) {$\ap{\fstsym}{\beta_{f(b)}} \ct \fst{\dpairsym(\dpathsym(s(b,a),d(b,a)))}$};
\node (N5)  at (0,1.6) {$\ap{\fstsym}{\beta_{f(b)}} \ct \ap{\fstsym}{\dpathsym(s(b,a),d(b,a))}$};
\node (N6)  at (0,0)   {$\ap{\fstsym}{\beta_{f(b)} \ct \dpathsym(s(b,a),d(b,a))}$};
\node (N7)  at (0,-1.6) {$\fst{\dpairsym\big(\beta_{f(b)} \ct \dpathsym(s(b,a),d(b,a))\big)}$};

\node (N8)  at (9,9.6) {$\alpha_{p(f\;b)} \ct \ap{\idfun{X}}{s(b,a)}$};
\node (N9)  at (9,8) {$\ap{\fstsym \circ u}{s(b,a)} \ct \alpha_{p(g\;b)}$};
\node (N10)  at (9,6.4)  {$\ap{\fstsym \circ u}{s(b,a)} \ct (\alpha_{p(g\;b)} \ct \reflsym)$};
\node (N11)  at (9,4.8) {$\ap{\fstsym \circ u}{s(b,a)} \ct \fst{\dpairsym(\beta_{g(b)})}$};
\node (N12)  at (9,3.2) {$\ap{\fstsym \circ u}{s(b,a)} \ct \ap{\fstsym}{\beta_{g(b)}}$};
\node (N13)  at (9,1.6) {$\ap{\fstsym}{\ap{u}{s(b,a)}} \ct \ap{\fstsym}{\beta_{g(b)}}$};
\node (N14)  at (9,0) {$\ap{\fstsym}{\ap{u}{s(b,a)} \ct \beta_{g(b)}}$};
\node (N15)  at (9,-1.6) {$\fst{\dpairsym\big(\ap{u}{s(b,a)} \ct \beta_{g(b)}\big)}$};

\draw[-] (N0) -- node[right]{} (N1);
\draw[-] (N1) -- node[right]{\scriptsize via \footnotesize $\eta_{f(b)}$} (N2);
\draw[-] (N2) -- node[right]{} (N3);
\draw[-] (N3) -- node[right]{} (N4);
\draw[-] (N4) -- node[right]{} (N5);
\draw[-] (N5) -- node[right]{} (N6);
\draw[-] (N6) -- node[right]{} (N7);

\draw[-] (N8) -- node[right]{\scriptsize via naturality of \footnotesize $\alpha$} (N9);
\draw[-] (N9) -- node[right]{} (N10);
\draw[-] (N10) -- node[right]{\scriptsize via \footnotesize $\eta_{g(b)}$} (N11);
\draw[-] (N11) -- node[right]{} (N12);
\draw[-] (N12) -- node[right]{} (N13);
\draw[-] (N13) -- node[right]{} (N14);
\draw[-] (N14) -- node[right]{} (N15);

\draw[-] (N0) -- node[right]{} (N8);
\draw[-] (N6) -- node[above]{\scriptsize via $(\star\star)$} (N14);
\draw[-] (N7) -- node[below]{\scriptsize via $(\star\star)$} (N15);
\end{tikzpicture}
\end{center}
The lower rectangle clearly commutes; it thus suffices to show that $(\star\star\star)$ commutes. 

The coherence condition - the last part of the algebra 2-cell between $(\fstsym\circ u,\gamma,\Theta)$, $(\idfun{X},\delta,\Phi)$ obtained from the uniqueness rule - now tells us that the following diagram commutes:
\begin{center}
\begin{tikzpicture}
\node (N0) at (4,2.4) {=};
\node (N1) at (0,4.8) {$\alpha_{p(f \; b)} \ct \ap{\idfun{X}}{s(b,a)}$};
\node (N2) at (8,4.8) {$\ap{\fstsym \circ u}{s(b,a)} \ct \alpha_{p(g\;b)}$};
\node (N3) at (0,3.2) {$(\fst{\dpairsym(\beta_{f(b)})} \ct \reflsym) \ct \ap{\idfun{X}}{s(b,a)}$};
\node (N4) at (8,3.2) {$\ap{\fstsym \circ u}{s(b,a)} \ct (\fst{\dpairsym(\beta_{g(b)})} \ct \reflsym)$};
\node (N5) at (0,1.6) {$\fst{\dpairsym(\beta_{f(b)})} \ct (\reflsym \ct \ap{\idfun{X}}{s(b,a)})$};
\node (N6) at (8,1.6) {$(\ap{\fstsym \circ u}{s(a,b)} \ct \fst{\dpairsym(\beta_{g(b)})}) \ct \reflsym$};
\node (N7) at (0,0) {$\fst{\dpairsym(\beta_{g(b)})} \ct (s(b,a) \ct \reflsym)$};
\node (N8) at (8,0) {$(\fst{\dpairsym(\beta_{g(b)})} \ct s(b,a)) \ct \reflsym$};
\draw[-] (N1) -- node[above]{\scriptsize \emph{naturality of} $\alpha$} (N2);
\draw[-] (N1) -- node[left]{\scriptsize \emph{via} $\invtri(\eta_{f(b)})$} (N3);
\draw[-] (N2) -- node[right]{\scriptsize \emph{via} $\invtri(\eta_{g(b)})$} (N4);
\draw[-] (N3) -- node[above]{\footnotesize} (N5);
\draw[-] (N4) -- node[above]{\footnotesize} (N6);
\draw[-] (N5) -- node[left]{\scriptsize \emph{via} $\invsq(\Phi(a,b))$} (N7);
\draw[-] (N6) -- node[right]{\scriptsize \emph{via} $\Theta(a,b)$} (N8);
\draw[-] (N7) -- node[above]{\footnotesize} (N8);
\end{tikzpicture}
\end{center}
After some expansion and simplification we get
\begin{center}
\begin{tikzpicture}
\node (N0) at (4,4) {=};
\node (N1) at (0,8) {$\alpha_{p(f \; b)} \ct \ap{\idfun{X}}{s(b,a)}$};
\node (N3) at (0,6.4) {$(\fst{\dpairsym(\beta_{f(b)})} \ct \reflsym) \ct \ap{\idfun{X}}{s(b,a)}$};
\node (N5) at (0,4.8) {$\fst{\dpairsym(\beta_{f(b)})} \ct (\reflsym \ct \ap{\idfun{X}}{s(b,a)})$};
\node[red] (N5b) at (0,3.2) {$\fst{\dpairsym(\beta_{f(b)})} \ct \ap{\idfun{X}}{s(b,a)}$};
\node[red] (N5c) at (0,1.6) {$\fst{\dpairsym(\beta_{f(b)})} \ct s(b,a)$};
\node (N7) at (0,0) {$\fst{\dpairsym(\beta_{g(b)})} \ct (s(b,a) \ct \reflsym)$};

\node (N2) at (8,8) {$\ap{\fstsym \circ u}{s(b,a)} \ct \alpha_{p(g\;b)}$};
\node (N4) at (8,6.4) {$\ap{\fstsym \circ u}{s(b,a)} \ct (\fst{\dpairsym(\beta_{g(b)})} \ct \reflsym)$};
\node (N6) at (8,4.8) {$(\ap{\fstsym \circ u}{s(a,b)} \ct \fst{\dpairsym(\beta_{g(b)})}) \ct \reflsym$};
\node[red] (N6b) at (8,3.2) {$\ap{\fstsym \circ u}{s(a,b)} \ct \fst{\dpairsym(\beta_{g(b)})}$};
\node[red] (N8a) at (8,1.6) {$\fst{\dpairsym(\beta_{g(b)})} \ct s(b,a)$};
\node (N8) at (8,0) {$(\fst{\dpairsym(\beta_{g(b)})} \ct s(b,a)) \ct \reflsym$};

\draw[-] (N1) -- node[above]{\scriptsize \emph{naturality of} $\alpha$} (N2);
\draw[-] (N1) -- node[left]{\scriptsize \emph{via} $\invtri(\eta_{f(b)})$} (N3);
\draw[-] (N2) -- node[right]{\scriptsize \emph{via} $\invtri(\eta_{g(b)})$} (N4);
\draw[-] (N3) -- node[above]{\footnotesize} (N5);
\draw[-] (N4) -- node[above]{\footnotesize} (N6);
\draw[red,-] (N5) -- node[left]{} (N5b);
\draw[red,-] (N5b) -- node[left]{} (N5c);
\draw[red,-] (N5c) -- node[left]{} (N7);
\draw[red,-] (N6) -- node[right]{} (N6b);
\draw[red,-] (N6b) -- node[right]{\scriptsize $\Theta(a,b)$} (N8a);
\draw[red,-] (N8a) -- node[right]{} (N8);
\draw[-] (N7) -- node[above]{\footnotesize} (N8);
\end{tikzpicture}
\end{center}
This is equivalent to
\begin{center}
\begin{tikzpicture}
\node (N0) at (4,4) {=};
\node (N1) at (0,8) {$\alpha_{p(f \; b)} \ct \ap{\idfun{X}}{s(b,a)}$};
\node[red] (N3) at (0,6.4) {$\alpha_{p(f \; b)} \ct s(b,a)$};
\node[red] (N5) at (0,4.8) {$(\fst{\dpairsym(\beta_{f(b)})} \ct \reflsym) \ct s(b,a)$};
\node[red] (N5b) at (0,3.2) {$\fst{\dpairsym(\beta_{f(b)})} \ct (\reflsym \ct s(b,a))$};
\node (N5c) at (0,1.6) {$\fst{\dpairsym(\beta_{f(b)})} \ct s(b,a)$};
\node (N7) at (0,0) {$\fst{\dpairsym(\beta_{g(b)})} \ct (s(b,a) \ct \reflsym)$};

\node (N2) at (8,8) {$\ap{\fstsym \circ u}{s(b,a)} \ct \alpha_{p(g\;b)}$};
\node (N4) at (8,6.4) {$\ap{\fstsym \circ u}{s(b,a)} \ct (\fst{\dpairsym(\beta_{g(b)})} \ct \reflsym)$};
\node (N6) at (8,4.8) {$(\ap{\fstsym \circ u}{s(a,b)} \ct \fst{\dpairsym(\beta_{g(b)})}) \ct \reflsym$};
\node (N6b) at (8,3.2) {$\ap{\fstsym \circ u}{s(a,b)} \ct \fst{\dpairsym(\beta_{g(b)})}$};
\node (N8a) at (8,1.6) {$\fst{\dpairsym(\beta_{g(b)})} \ct s(b,a)$};
\node (N8) at (8,0) {$(\fst{\dpairsym(\beta_{g(b)})} \ct s(b,a)) \ct \reflsym$};

\draw[-] (N1) -- node[above]{\scriptsize \emph{naturality of} $\alpha$} (N2);
\draw[red,-] (N1) -- node[left]{} (N3);
\draw[red,-] (N3) -- node[left]{\scriptsize \emph{via} $\invtri(\eta_{f(b)})$} (N5);
\draw[-] (N2) -- node[right]{\scriptsize \emph{via} $\invtri(\eta_{g(b)})$} (N4);
\draw[-] (N4) -- node[above]{\footnotesize} (N6);
\draw[red,-] (N5) -- node[left]{} (N5b);
\draw[red,-] (N5b) -- node[left]{} (N5c);
\draw[-] (N5c) -- node[left]{} (N7);
\draw[-] (N6) -- node[right]{} (N6b);
\draw[-] (N6b) -- node[right]{\scriptsize $\Theta(a,b)$} (N8a);
\draw[-] (N8a) -- node[right]{} (N8);
\draw[-] (N7) -- node[above]{\footnotesize} (N8);
\end{tikzpicture}
\end{center}
After some cleanup we get
\begin{center}
\begin{tikzpicture}
\node (N0) at (4,4.5) {=};
\node (N1) at (0,8) {$\alpha_{p(f \; b)} \ct \ap{\idfun{X}}{s(b,a)}$};
\node (N3) at (0,6.4) {$\alpha_{p(f \; b)} \ct s(b,a)$};
\node (N5) at (0,4.8) {$(\fst{\dpairsym(\beta_{f(b)})} \ct \reflsym) \ct s(b,a)$};
\node (N5b) at (0,3.2) {$\fst{\dpairsym(\beta_{f(b)})} \ct (\reflsym \ct s(b,a))$};
\node (N5c) at (4,1) {$\fst{\dpairsym(\beta_{f(b)})} \ct s(b,a)$};

\node (N2) at (8,8) {$\ap{\fstsym \circ u}{s(b,a)} \ct \alpha_{p(g\;b)}$};
\node (N4) at (8,6.4) {$\ap{\fstsym \circ u}{s(b,a)} \ct (\fst{\dpairsym(\beta_{g(b)})} \ct \reflsym)$};
\node (N6) at (8,4.8) {$(\ap{\fstsym \circ u}{s(a,b)} \ct \fst{\dpairsym(\beta_{g(b)})}) \ct \reflsym$};
\node (N6b) at (8,3.2) {$\ap{\fstsym \circ u}{s(a,b)} \ct \fst{\dpairsym(\beta_{g(b)})}$};

\draw[-] (N1) -- node[above]{\scriptsize \emph{naturality of} $\alpha$} (N2);
\draw[-] (N1) -- node[left]{} (N3);
\draw[-] (N3) -- node[left]{\scriptsize \emph{via} $\invtri(\eta_{f(b)})$} (N5);
\draw[-] (N2) -- node[right]{\scriptsize \emph{via} $\invtri(\eta_{g(b)})$} (N4);
\draw[-] (N4) -- node[above]{\footnotesize} (N6);
\draw[-] (N5) -- node[left]{} (N5b);
\draw[-] (N5b) -- node[left]{} (N5c);
\draw[-] (N6) -- node[right]{} (N6b);
\draw[-] (N6b) -- node[right]{\;\;\scriptsize $\Theta(a,b)$} (N5c);
\end{tikzpicture}
\end{center}
Further expansion yields
\begin{center}
\begin{tikzpicture}
\node (N0) at (4,2.9) {=};
\node (N1) at (0,8) {$\alpha_{p(f \; b)} \ct \ap{\idfun{X}}{s(b,a)}$};
\node (N3) at (0,6.4) {$\alpha_{p(f \; b)} \ct s(b,a)$};
\node[red] (N5) at (0,4.8) {$(\alpha_{p(f \; b)} \ct \reflsym) \ct s(b,a)$};
\node[red] (N5aa) at (0,3.2) {$\fst{\dpairsym(\beta_{f(b)})} \ct s(b,a)$};
\node (N5ab) at (0,1.6) {$(\fst{\dpairsym(\beta_{f(b)})} \ct \reflsym) \ct s(b,a)$};
\node (N5b) at (0,0) {$\fst{\dpairsym(\beta_{f(b)})} \ct (\reflsym \ct s(b,a))$};
\node (N5c) at (4,-2.2) {$\fst{\dpairsym(\beta_{f(b)})} \ct s(b,a)$};

\node (N2) at (8,8) {$\ap{\fstsym \circ u}{s(b,a)} \ct \alpha_{p(g\;b)}$};
\node[red] (N2aa) at (8,6.4) {$\ap{\fstsym \circ u}{s(b,a)} \ct (\alpha_{p(g\;b)} \ct \reflsym)$};
\node[red] (N2ab) at (8,4.8) {$\ap{\fstsym \circ u}{s(b,a)} \ct \fst{\dpairsym(\beta_{g(b)})}$};
\node (N4) at (8,3.2) {$\ap{\fstsym \circ u}{s(b,a)} \ct (\fst{\dpairsym(\beta_{g(b)})} \ct \reflsym)$};
\node (N6) at (8,1.6) {$(\ap{\fstsym \circ u}{s(a,b)} \ct \fst{\dpairsym(\beta_{g(b)})}) \ct \reflsym$};
\node (N6b) at (8,0) {$\ap{\fstsym \circ u}{s(a,b)} \ct \fst{\dpairsym(\beta_{g(b)})}$};

\draw[-] (N1) -- node[above]{\scriptsize \emph{naturality of} $\alpha$} (N2);
\draw[-] (N1) -- node[left]{} (N3);
\draw[red,-] (N3) -- node[left]{} (N5);
\draw[red,-] (N2) -- node[left]{} (N2aa);
\draw[red,-] (N2aa) -- node[right]{\scriptsize \emph{via} $\invtri(\eta_{g(b)})$} (N2ab);
\draw[red,-] (N2ab) -- node[left]{} (N4);
\draw[-] (N4) -- node[above]{\footnotesize} (N6);
\draw[red,-] (N5) -- node[left]{\scriptsize \emph{via} $\invtri(\eta_{f(b)})$} (N5aa);
\draw[red,-] (N5aa) -- node[left]{} (N5ab);
\draw[-] (N5ab) -- node[left]{} (N5b);
\draw[-] (N5b) -- node[left]{} (N5c);
\draw[-] (N6) -- node[right]{} (N6b);
\draw[-] (N6b) -- node[right]{\;\;\scriptsize $\Theta(a,b)$} (N5c);
\end{tikzpicture}
\end{center}
A final cleanup yields
\begin{center}
\begin{tikzpicture}
\node (N0) at (4,5.3) {=};
\node (N1) at (0,8) {$\alpha_{p(f \; b)} \ct \ap{\idfun{X}}{s(b,a)}$};
\node (N3) at (0,6.4) {$\alpha_{p(f \; b)} \ct s(b,a)$};
\node (N5) at (0,4.8) {$(\alpha_{p(f \; b)} \ct \reflsym) \ct s(b,a)$};
\node (N5aa) at (4,2.6) {$\fst{\dpairsym(\beta_{f(b)})} \ct s(b,a)$};

\node (N2) at (8,8) {$\ap{\fstsym \circ u}{s(b,a)} \ct \alpha_{p(g\;b)}$};
\node (N2aa) at (8,6.4) {$\ap{\fstsym \circ u}{s(b,a)} \ct (\alpha_{p(g\;b)} \ct \reflsym)$};
\node (N2ab) at (8,4.8) {$\ap{\fstsym \circ u}{s(b,a)} \ct \fst{\dpairsym(\beta_{g(b)})}$};

\draw[-] (N1) -- node[above]{\scriptsize \emph{naturality of} $\alpha$} (N2);
\draw[-] (N1) -- node[left]{} (N3);
\draw[-] (N3) -- node[left]{} (N5);
\draw[-] (N2) -- node[left]{} (N2aa);
\draw[-] (N2aa) -- node[right]{\scriptsize \emph{via} $\eta_{g(b)}$} (N2ab);
\draw[-] (N2ab) -- node[right]{\;\;\scriptsize $\Theta(a,b)$} (N5aa);
\draw[-] (N5) -- node[left]{\scriptsize \emph{via} $\eta_{f(b)}$\;\;} (N5aa);
\end{tikzpicture}
\end{center}
which is precisely the diagram $(\star\star\star)$.

\end{document}